\documentclass[11pt]{article}
\usepackage{microtype}%if unwanted, comment out or use option "draft"

\usepackage{amsmath}
\usepackage{amsfonts}
\usepackage{amssymb, amsthm, epsfig}
\usepackage{mathrsfs}

\usepackage{graphicx}
\usepackage{verbatim}
\usepackage{color}
\usepackage{epstopdf}

\usepackage{mathtools}
\usepackage{float}
\usepackage{algorithm}
\usepackage[noend]{algorithmic}

\usepackage{xspace}
\usepackage{color}

\usepackage{wrapfig}
\usepackage{bbm}
\usepackage{tabularx}
\usepackage{footnote}
\usepackage{thm-restate}

\newcommand{\eps}{\varepsilon}

\newcommand{\R}{\ensuremath{\mathbb{R}}}

\newcommand{\etal}{\emph{et al.}\xspace}

\newcommand{\prob}[1]{\ensuremath{\textbf{{\sffamily Pr}}\hspace{-.8mm}\left[#1\right]}}

\newcommand{\norm}[1]{\lVert #1 \rVert}

\newcommand\numberthis{\addtocounter{equation}{1}\tag{\theequation}}
\newcommand{\abs}[1]{| #1 |}
\newcommand{\setdef}[2]{\{ #1 \mid #2 \}}
\newcommand{\inner}[2]{\langle #1,#2 \rangle}

\newcommand{\Gr}{\textsf{Grid}}

\newcommand{\Gb}{\overline{\mathcal{G}}}
\newcommand{\ilog}{\textup{ilog}}

\newtheorem{theorem}{Theorem}

\title{Optimal Coreset for Gaussian Kernel Density Estimation}

\author{Wai Ming Tai \\ University of Chicago}

\date{ }

\begin{document}

\maketitle

\begin{abstract}
	Given a point set $P\subset \mathbb{R}^d$, the kernel density estimate of $P$ is defined as 
	\[
	\overline{\mathcal{G}}_P(x) = \frac{1}{\left|P\right|}\sum_{p\in P}e^{-\left\lVert x-p \right\rVert^2}
	\] 
	for any $x\in\mathbb{R}^d$.
	We study how to construct a small subset $Q$ of $P$ such that the kernel density estimate of $P$ is approximated by the kernel density estimate of $Q$.
	This subset $Q$ is called a coreset.
	The main technique in this work is constructing a $\pm 1$ coloring on the point set $P$ by discrepancy theory and we leverage Banaszczyk's Theorem.
	When $d>1$ is a constant, our construction gives a coreset of size $O\left(\frac{1}{\varepsilon}\right)$ as opposed to the best-known result of $O\left(\frac{1}{\varepsilon}\sqrt{\log\frac{1}{\varepsilon}}\right)$.
	It is the first result to give a breakthrough on the barrier of $\sqrt{\log}$ factor even when $d=2$. 
	
\end{abstract}

\section{Introduction}

Kernel density estimation is a non-parametric way to estimate a probability distribution.
Given a point set $P\subset\mathbb{R}^d$, the kernel density estimate (KDE) of $P$ smooths out $P$ to a continuous function \cite{scott2015multivariate,silverman1986density}.
More precisely, given a point set $P\subset\mathbb{R}^d$ and a kernel $K:\mathbb{R}^d\times\mathbb{R}^d \rightarrow \mathbb{R}$, KDE is defined as the function $\Gb_P(x) = \frac{1}{\abs{P}}\sum_{p\in P} K(x,p)$ for any $x\in\mathbb{R}^d$.
Here, the point $x$ is called a \emph{query}.
One common example of kernel $K$ is the Gaussian kernel, which is $K(x,y) = e^{-\norm{x-y}^2}$ for any $x,y\in\mathbb{R}^d$, and it is the main focus of this paper.
A wide range of application includes outlier detection \cite{zou2014unsupervised}, clustering \cite{rinaldo2010generalized}, topological data analysis \cite{phillips2015geometric,chazal2017robust}, spatial anomaly detection \cite{agarwal2014union,han2019kernel}, statistical hypothesis test \cite{gretton2012kernel} and other \cite{jeff2020gaussiansketch,lee2021finding}.

Generally speaking, the techniques using kernels are called \emph{kernel methods}, in which KDE is the central role in these techniques.
Kernel methods are prevalent in machine learning and statistics and often involve optimization problems.
Optimization problems are generally hard in the sense that solving them usually has a super-linear or even  an exponential dependence on the input's size in its running time.
Therefore, reducing the size of the input will be desirable.
A straightforward way to achieve this is extracting a small subset $Q$ of the input $P$.
This paper will study the construction of the subset $Q$ such that $\Gb_Q$ approximates $\Gb_P$.

Classically, statisticians concern about different types of error such as $L_1$-error \cite{DG84} or $L_2$-error \cite{scott2015multivariate,silverman1986density}.
However, there are multiple modern applications that require $L_\infty$-error such as preserving classification margin \cite{scholkopf2002learning}, density estimation \cite{zheng2015error}, topology \cite{phillips2015geometric} and hypothesis test on distributions \cite{gretton2012kernel}.
For example, in topological data analysis, we might want to study the persistent homology of a super-level set of a kernel density estimate. 
In this case, $L_\infty$-error plays an important role here since a small perturbation could cause a significant change in its persistence diagram.
Formally, we would like to solve the following problem.
\begin{flushleft}
	\emph{Given a point set $P\subset \mathbb{R}^d$ and $\eps>0$, we construct a subset $Q$ of $P$ such that 
		\[
		\sup_{x\in\mathbb{R}^d} \abs{\Gb_P(x) - \Gb_Q(x)} = \sup_{x\in\mathbb{R}^d} \abs{\frac{1}{\abs{P}}\sum_{p\in P}e^{-\norm{x-p}^2} - \frac{1}{\abs{Q}}\sum_{q\in Q}e^{-\norm{x-q}^2}} \leq \eps.
		\]
		Then, how small can the size of $Q$, $\abs{Q}$, be?}
\end{flushleft}
We call this subset $Q$ an \emph{$\eps$-coreset}.

\subsection{Known results}

We now discuss some previous results for the size of an $\eps$-coreset.
The summary is presented in Table \ref{tbl:compare}.

\begin{savenotes}
	\begin{table}[!h]
		\begin{center}
			\begin{tabular}{|r|c|c|}
				\hline
				Paper & Coreset Size & $d$ \\
				\hline 
				\multicolumn{3}{|l|}{Previous Results} \\
				\hline
				Joshi \etal~\cite{joshi2011comparing} & $O (d/\eps^2 )$ &  any \\
				Lopaz-Paz \etal \cite{lopez2015towards} & $O (1/\eps^2 )$ & any \\
				Lacoste-Julien \etal \cite{lacoste2015sequential} & $O (1/\eps^2 )$ & any \\
				Joshi \etal~\cite{joshi2011comparing} & $O (1/\eps )$ &  1 \\
				Joshi \etal \cite{joshi2011comparing} & sub-$O (1/\eps^2 )$ & constant \\
				Phillips \cite{phillips2013varepsilon} & $O ( ((1/\eps^2)\log\frac{1}{\eps} )^{\frac{d}{d+2}} )$ & constant \\
				Phillips and Tai \cite{phillips2018improved} & $O ((1/\eps)\log^d \frac{1}{\eps} )$ & constant \\
				Phillips and Tai \cite{phillips2019near} & $O ((\sqrt{d}/\eps)\sqrt{\log \frac{1}{\eps}} )$ & any \\
				\hline
				Phillips \cite{phillips2013varepsilon} & $\Omega (1/\eps )$ &  any\\
				Phillips and Tai \cite{phillips2018improved} & $\Omega (1/\eps^2 )$ &  $\geq\frac{1}{\eps^2}$\\
				Phillips and Tai \cite{phillips2019near} & $\Omega (\sqrt{d}/\eps )$ & $\leq\frac{1}{\eps^2}$ \\
				\hline
				\multicolumn{3}{|l|}{Closely-related Result} \\
				\hline 
				Karnin and Liberty \cite{karnin2019discrepancy} & $O (\sqrt{d}/\eps )$\footnote{This result assumes that both data points and queries lie inside a constant region} & any \\
				\hline
				\multicolumn{3}{|l|}{Our Result} \\
				\hline
				\textbf{This paper} & $O (1/\eps )$ & constant \\
				\hline
			\end{tabular} 
		\end{center}
		\caption{Asymptotic $\eps$-coreset sizes in terms of $\eps$ and $d$.
			\label{tbl:compare}}
	\end{table}
\end{savenotes}

Josh \etal \cite{joshi2011comparing} showed that random sampling can achieve the size of $O(\frac{d}{\eps^2})$.
They investigated the VC-dimension of the super-level sets of a kernel and analyzed that the sample size can be bounded by it.
In particular, the super-level sets of the Gaussian kernel are balls in $\mathbb{R}^d$.
It reduces the problem to bounding the sample size of the range space of balls.

Lopaz-Paz \etal \cite{lopez2015towards} later proved that the size of the coreset can be reduced to $O(\frac{1}{\eps^2})$ by random sampling.
They studied the reproducing kernel Hilbert space (RKHS) associated with a positive-definite kernel \cite{aronszajn1950theory,wahba1999support,sriperumbudur2010hilbert}.
Note that the Gaussian kernel is a positive-definite kernel.
In RKHS, one can bound the $L_\infty$-error between two KDEs of point sets $P$ and $Q$ by the kernel distance of $P$ and $Q$.
They showed that the sample size of $O(\frac{1}{\eps^2})$ is sufficient to bound the kernel distance.

Other than random sampling, Lacoste-Julien \etal \cite{lacoste2015sequential} showed a greedy approach can also achieve the size of $O(\frac{1}{\eps^2})$.
They applied Frank-Wolfe algorithm \cite{clarkson2010coresets,gartner2009coresets} in RKHS to bound the error of the kernel distance.

Note that all of the above results have a factor of $\frac{1}{\eps^2}$.
Josh \etal \cite{joshi2011comparing} first showed that a sub-$O(\frac{1}{\eps^2})$ result can be obtained by reducing the problem to constructing an $\eps$-approximation for the range space of balls \cite{matousek2009geometric}.
They assumed that $d$ is constant.
For the case of $d=1$, their result gives the size of $O(\frac{1}{\eps})$.

Later, Phillips \cite{phillips2013varepsilon} improved the result to $O ( (\frac{1}{\eps^2}\log\frac{1}{\eps} )^{\frac{d}{d+2}} )$ for constant $d$ via geometric matching.
It is based on the discrepancy approach.
Namely, they construct a $\pm 1$ coloring on the point set, recursively drop the points colored $-1$ and construct another $\pm 1$ coloring on the points colored $+1$.
We will discuss it in more detail below.
Notably, for the case of $d=2$, their bound is $O(\frac{1}{\eps}\sqrt{\log \frac{1}{\eps}})$ which is nearly-optimal (as a preview, the optimal bound is $\Omega(\frac{1}{\eps})$) and is the first nearly-linear result for the case of $d>1$.

Recently, Phillips and Tai \cite{phillips2018improved} further improved the size of a coreset to $O(\frac{1}{\eps}\log^d \frac{1}{\eps})$ for constant $d$.
It is also based on the discrepancy approach.
They exploited the fact that the Gaussian kernel is multiplicatively separable.
It implies that the Gaussian kernel can be rewritten as the weighted average of a family of axis-parallel boxes in $\mathbb{R}^d$.
Finally, they reduced the problem to Tusn{\'a}dy's problem \cite{bansal2017algorithmic,aistleitner2016tusnady}.

Also, Phillips and Tai \cite{phillips2019near} proved a nearly-optimal result of $O(\frac{\sqrt{d}}{\eps}\sqrt{\log \frac{1}{\eps}})$ shortly after that.
It is also based on the discrepancy approach.
They observed that the underlying structure of the positive-definite kernel allows us to bound the norm of the vectors and apply the lemma in \cite{matouvsek2020factorization}, which used Banaszczyk's Theorem \cite{banaszczyk1998balancing,bansal2018gram}.
Recall that the Gaussian kernel is a positive-definite kernel.

Except for the upper bound, there are some results on the lower bound for the size of an $\eps$-coreset.
Phillips \cite{phillips2013varepsilon} provided the first lower bound for the size of a coreset.
They proved a lower bound of $\Omega(\frac{1}{\eps})$ by giving an example that all points are spread out.
When assuming $d>\frac{1}{\eps^2}$, Phillips and Tai  \cite{phillips2018improved} gave another example that forms a simplex and showed a lower bound of $\Omega(\frac{1}{\eps^2})$.
Later, Phillips and Tai \cite{phillips2019near} combined the techniques of the above two results and showed the lower bound of $\Omega(\frac{\sqrt{d}}{\eps})$.

There are other conditional bounds for this problem.
We suggest the readers refer to \cite{phillips2019near} for a more extensive review.
Recently, Karnin and Liberty \cite{karnin2019discrepancy} defined the notion of Class Discrepancy which governs the coreset-complexity of different families of functions. 
Specifically, for analytic functions of squared distances (such as the Gaussian kernel), their analysis gives a discrepancy bound $D_m = O(\frac{\sqrt{d}}{m})$ which gives a coreset of size $O(\frac{\sqrt{d}}{\eps})$. 
Their approach also used the discrepancy technique or, more precisely, Banaszczyk's Theorem \cite{banaszczyk1998balancing,bansal2018gram}.
Unfortunately, their analysis requires \emph{both} the point set $P$ and the query $x$ lie in a ball of a fixed radius $R$.
Therefore, their result has a dependence on $R$.
Strictly speaking, their result is not comparable to ours.
It is not clear how to remove this assumption of $R$ based on their result.
Also, the lower bound constructions in \cite{phillips2013varepsilon,phillips2019near} rely on the fact that $P$ is in an unbounded region and hence it is not clear how their result is comparable to the existing lower results.

\subsection{Related works}

In computational geometry, an $\eps$-approximation is the approximation of a general set by a smaller subset.
Given a set $S$ and a collection $\mathcal{C}$ of subsets of $S$, a subset $A\subset S$ is called an $\eps$-approximation if $\abs{\frac{\abs{T}}{\abs{S}} - \frac{\abs{T\cap A}}{\abs{A}}}\leq \eps$ for all $T\in \mathcal{C}$.
The pair $(S,\mathcal{C})$ is called a set system (also known as a range space or a hypergraph).
One can rewrite the above guarantee as $\abs{\frac{1}{\abs{S}}\sum_{x\in S}\mathbbm{1}_T(x) - \frac{1}{\abs{A}}\sum_{x\in A}\mathbbm{1}_{T}(x)}\leq \eps$ where $\mathbbm{1}_T$ is the indicator function of set $T$.
If we replace this indicator function by a kernel such as the Gaussian kernel, it is the same as our $\eps$-coreset.
There is a rich history on the construction of an $\eps$-approximation \cite{chazelle2001discrepancy,matousek2009geometric}.
One notable method is discrepancy theory, which is also our main technique.
There is a wide range of techniques employed in this field.
In the early 1980s, Beck devised the technique of partial coloring \cite{beck1981roth}, and later a refinement of this technique called entropy method was introduced by Spencer \cite{spencer1985six}.
The entropy method is first used to solve the famous "six standard deviations" theorem: given a set system of $n$ points and $n$ subsets, there is a coloring of discrepancy at most $6\sqrt{n}$. In contrast, random coloring gives the discrepancy of $O(\sqrt{n\log n})$.
%However, this entropy method is a non-constructive approach.
%Namely, they did not provide an efficient algorithm to find the coloring achieving such bound.
A more geometric example in discrepancy theory is Tusn{\'a}dy's problem.
It states that, given point set $P$ of size $n$ in $\R^d$, construct a $\pm 1$ coloring $\sigma$ on $P$ such that the discrepancy $\min_{\sigma}\max_{R}\abs{\sum_{P\cap R}\sigma(p)}$ is minimized where $\max_R$ is over all axis-parallel boxes $R$.
One previous approach of our $\eps$-coreset problem reduces the problem to Tusn{\'a}dy's problem.

On the topic of approximating KDE, Fast Gauss Transform \cite{greengard1991fast} is a method to preprocess the input point set such that the computation of KDE at a query is faster than the brute-force approach.
The idea in this method is expanding the Gaussian kernel by Hermite polynomials and truncating the expansion.
Assuming that the data set lies inside a bounded region, the query time in this method is poly-logarithmic of $n$ for constant dimension $d$.
Also, Charikar \etal~\cite{charikar2020kernel} studied the problem of designing a data structure that preprocesses the input to answer a KDE query in a faster time.
They used locality-sensitive hashing to perform their data structure.
However, the guarantee they obtained is a relative error, while ours is an additive error.
More precisely, given a point set $P\subset \mathbb{R}^d$, Charikar \etal designed a data structure such that, for any query $x'\in\mathbb{R}^d$, the algorithm answers the value $\Gb_P(x') = \sum_{p\in P} e^{-\norm{x'-p}^2}$ within $(1+\eps)$-relative error.
Also, the query time of their data structure is sublinear of $n$.

\subsection{Our result}

\begin{wraptable}{r}{3.3in}
	\begin{tabular}{c | c  c c}
		\hline 
		$d$ & Upper & Lower & 
		\\ \hline
		$1$ & $1/\eps$ & $1/\eps$ & \cite{joshi2011comparing,phillips2013varepsilon}
		\\
		constant & $1/\eps$ &  &  \textbf{new}
		\\
		any & $\sqrt{d}/\eps \cdot \sqrt{\log \frac{1}{\eps}}$ & $\sqrt{d}/\eps$ & \cite{phillips2019near}
		\\
		$\geq \frac{1}{\eps^2}$ & $1/\eps^2$ & $1/\eps^2$ & \cite{bach2012equivalence,phillips2018improved}
		\\ \hline
	\end{tabular}
	\caption{\label{tbl:results} Size bounds of $\eps$-coresets for the Gaussian kernel}
\end{wraptable}

We construct an $\eps$-coreset and bound the size of the $\eps$-coreset via discrepancy theory.
Roughly speaking, we construct a $\pm 1$ coloring on our point set such that its discrepancy is small.
Then, we drop the points colored $-1$ and recursively construct a $\pm 1$ coloring on the points colored $+1$.
Eventually, the remaining point set is the desired coreset.
A famous theorem in discrepancy theory is Banaszczyk's Theorem \cite{banaszczyk1998balancing,bansal2018gram}.
We will use Banaszczyk's Theorem to construct a coloring and prove the discrepancy is small by induction.
To the best of our knowledge, this induction analysis combining with Banaszczyk's Theorem has not been seen in discrepancy theory before.
In the constant dimensional space, we carefully study the structure of the Gaussian kernel and it allows us to construct an $\eps$-coreset of size $O(1/\eps)$.
Our result is the first result to break the barrier of $\sqrt{\log}$ factor even when $d=2$.
The summary of the best-known result is shown in Table \ref{tbl:results}.

\begin{restatable}{theorem}{mainthm}\label{thm:main_g}
	Suppose $P\subset \mathbb{R}^d$ a point set of size $n$.
	Let $\Gb_P$ be the Gaussian kernel density estimate of $P$, i.e. $\Gb_P(x) = \frac{1}{\abs{P}}\sum_{p\in P} e^{-\norm{x-p}^2}$ for any $x\in\mathbb{R}^d$.
	For a fixed constant $d$, there is an algorithm that constructs a subset $Q\subset P$ of size $O (\frac{1}{\eps} )$ such that $\sup_{x\in\mathbb{R}^d}\abs{\Gb_P(x) - \Gb_Q(x)} < \eps$ and has a polynomial running time in $n$.
\end{restatable}

Even if $d=1$, the best known result is $O(1/\eps)$ by \cite{joshi2011comparing,phillips2018improved}, which is optimal.
Their approach is to reduce the problem to Tusn{\'a}dy's problem.
A trivial solution of Tusn{\'a}dy's problem (and hence our problem) is: sort $P$ and assign $\pm 1$ on each point alternately.
However, it is not clear that how this simple solution can be generalized to the higher dimensional case.
Our algorithm gives a non-trivial perspective even though the optimal result was achieved previously.

\section{Preliminaries}

Our approach for constructing a coreset relies on discrepancy theory, which is a similar technique in range counting coreset \cite{chazelle1996linear,phillips2008algorithms,bentley1980decomposable}.
We first introduce an equivalent problem (up to a constant factor) as follows.
\begin{flushleft}
	\emph{Given a point set $P\subset \mathbb{R}^d$, what is the smallest quantity of $\sup_{x\in\mathbb{R}^d}\abs{\sum_{p\in P}\sigma(p)e^{-\norm{x-p}^2}}$ over all $\sigma$ in the set of colorings from $P$ to $\{-1,+1\}$?}
\end{flushleft}
Now, one can intuitively view the equivalence in the following way.
If we rewrite the objective as:
\[
\frac{1}{\abs{P}}\abs{\sum_{p\in P}\sigma(p)e^{-\norm{x-p}^2}} = \abs{\frac{1}{\abs{P}}\sum_{p\in P}e^{-\norm{x-p}^2} - \frac{1}{\abs{P}/2}\sum_{p\in P_+}e^{-\norm{x-p}^2}}
\]
where $P_+\subset P$ is the set of points that is assigned $+1$, then we can apply the halving technique \cite{chazelle1996linear,phillips2008algorithms} which recursively invokes the coloring algorithm and retains the points assigned $+1$ until the subset of the desired size remains.
Note that there is no guarantee that half of the points are assigned $+1$, while the other half is assigned $-1$.
However, we can handle this issue by some standard techniques \cite{matousek2009geometric} or see our proof for details.

Also, we define the following notations.
Given a point set $P\subset \mathbb{R}^d$, a coloring $\sigma:P \rightarrow \{-1,+1\}$ and a point $x\in S$, we define the signed discrepancy $\mathcal{D}_{P,\sigma}(x)$ as 
\[
\mathcal{D}_{P,\sigma}(x) = \sum_{p\in P}\sigma(p)e^{-\norm{x-p}^2}
\]
It is worth noting that we expect $\abs{\mathcal{D}_{P,\sigma}(x)} < O(1)$ in order to construct an $\eps$-coreset of size $O(\frac{1}{\eps})$ via this halving technique.

An important result in discrepancy theory is Banaszczyk's Theorem  \cite{banaszczyk1998balancing}.
\begin{theorem}[Banaszczyk's Theorem \cite{banaszczyk1998balancing}]\label{thm:b_thm_orig}
	Suppose we are given a convex body $K\subset \mathbb{R}^m$ of the Gaussian measure at least $\frac{1}{2}$ and $n$ vectors $v^{(1)},v^{(2)},\dots,v^{(n)}\in \mathbb{R}^m$ of norm at most $1$, there is a coloring $\sigma:[n]  \rightarrow \{-1,+1\}$ such that the vector $\sum_{i=1}^n \sigma(i)v^{(i)} \in cK= \setdef{c\cdot y}{y\in K}$. Here, $c$ is an absolute constant and the Gaussian measure of a convex body $K$ is defined as $\int_{x\in K}\frac{1}{(2\pi)^{d/2}}e^{-\norm{x}^2/2}dx$.
\end{theorem}
The original proof of this theorem is non-constructive.
Bansal \etal \cite{bansal2018gram} proved that there is an efficient algorithm to construct the coloring in Banaszczyk's Theorem.
Moreover, assuming $m<n$, the running time is $O(n^{\omega+1})$ where $\omega$ is the exponent of matrix multiplication.
\begin{theorem}[Constructive version of Banaszczyk's Theorem \cite{bansal2018gram}]\label{thm:b_thm}
	Suppose we are given $n$ vectors $v^{(1)},\dots,v^{(n)}\in\mathbb{R}^m$ of norm at most $1$, there is an efficient randomized algorithm that constructs a coloring $\sigma$ on $P$ with the following guarantee: there are two absolute constants $C',C''$ such that, for any unit vector $\theta\in\mathbb{R}^m$ and $\alpha>0$, we have
	\[
	\prob{\abs{\inner{\theta}{X}}>\alpha} < C'e^{-C''\alpha^2}
	\]
	where $X$ is the random variable of $\sum_{i=1}^n\sigma(i)v^{(i)}$.
	The probability in the above statement is distributed over all $\pm 1$ colorings.
\end{theorem}

Finally, we introduce a useful theorem which is Markov Brother's Inequality.
\begin{theorem}[Markov Brother's Inequality \cite{markov1889question}]\label{lem:mb_ineq}
	Let $\mathcal{P}(x)$ be a polynomial of degree $\rho$.
	Then,
	\[
	\sup_{x\in[0,1]}\abs{\mathcal{P}'(x)} \leq 2\rho^2\sup_{x\in[0,1]}\abs{\mathcal{P}(x)}
	\]
	Here, $\mathcal{P}'$ is the derivative of $\mathcal{P}$.
\end{theorem}

\section{Proof overview}\label{sec:overview}

As we mentioned before, our equivalent problem statement suggests that we need to construct a $\pm 1$ coloring on the input point set such that the absolute value of the signed discrepancy at all points is small.
In this section, we will give an overview on how we construct the coloring and how it gives us the desired guarantees.

For exposition purposes, we illustrate the idea for the case of $d=1$ even though previous results \cite{joshi2011comparing,phillips2018improved} showed this case is trivial.
Recall that our problem definition is: given a point set $P\subset \mathbb{R}$ of size $n$, construct a $\pm 1$ coloring $\sigma$ on $P$ such that the absolute value of the signed discrepancy
\[
\abs{\mathcal{D}_{P,\sigma}(x)} = \abs{\sum_{p\in P}\sigma(p)e^{-(x-p)^2}}
\]
is bounded from above by a constant for \emph{all} $x\in\mathbb{R}$.

\subparagraph{Some general observations} 
We first make some observations.
Note that $\mathcal{D}_{P,\sigma}$ is a smooth function of $x$ that the slope at any $x$ is bounded.
It means that if $\abs{\mathcal{D}_{P,\sigma}(x_0)}$ is small for some point $x_0$ then $\abs{\mathcal{D}_{P,\sigma}(y)}$ is also small for any point $y$ at a neighborhood of $x_0$.
Another observation is that $\mathcal{D}_{P,\sigma}$ is basically a linear combination of Gaussians and hence $\abs{\mathcal{D}_{P,\sigma}(x)}$ is small for any $x$ that is far away from all points in $P$.

Combining these two observations, if we lay down a grid on $\mathbb{R}$ and consider the grid points that is not too far away from $P$, then we only need to construct a coloring $\sigma$ such that $\abs{\mathcal{D}_{P,\sigma}(x)}$ is small for all $x$ in a \emph{finite} set and it implies that $\abs{\mathcal{D}_{P,\sigma}(x)}$ is small for all $x\in\mathbb{R}$.
It is crucial because we preview that our algorithm for constructing the coloring $\sigma$ is a randomized algorithm and the size of the finite set controls the number of events when we apply the union bound.
Note that these observations hold for \emph{any} coloring.

\subparagraph{Techniques from the previous result}
Now, we make the above observations more quantitative.
Since the slope of each Gaussian at any point is bounded by $O(1)$ and there are $n$ Gaussians in $\mathcal{D}_{P,\sigma}$, by triangle inequality, the absolute value of the slope of $\mathcal{D}_{P,\sigma}$ at any point is bounded by $O(n)$.
Hence, if $\abs{\mathcal{D}_{P,\sigma}(x_0)}$ is bounded by $\alpha$ for any point $x_0$ for any $\alpha$ then $\abs{\mathcal{D}_{P,\sigma}(y)}$ is bounded by $\alpha+O(1)$ for all $y$ that $\abs{x_0-y}<O(1/n)$.
Also, Gaussians decay exponentially and hence $\abs{\mathcal{D}_{P,\sigma}(x)} <O(1)$ for any $x$ that $\abs{x-p}>\Omega(\sqrt{\log n})$ for all $p\in P$.

\emph{If} a coloring $\sigma$ satisfies that 
\[
\text{$\abs{\mathcal{D}_{P,\sigma}(x)} < \alpha$ for \emph{any} $x\in\mathbb{R}$ with probability at least $1-O(e^{-\Omega(\alpha^2)})$ for any $\alpha$} \numberthis\label{eq:property}
\]
then it implies, by union bound, this coloring $\sigma$ satisfies that 
\[
\text{$\abs{\mathcal{D}_{P,\sigma}(x)} < \alpha+O(1)$ for \emph{all} $x\in\mathbb{R}$ with probability at least $1-N\cdot O(e^{-\Omega(\alpha^2)})$}
\]
where $N$ is the number of grid points that are in the grid of cell width $\Omega(1/n)$ and lie around some point in $P$ within a radius of $O(\sqrt{\log n})$.
The number $N$ is bounded by $O(n^2\sqrt{\log n})$ because for each point $p\in P$ there are $O(\sqrt{\log n}/(1/n)) = O(n\sqrt{\log n})$ grid points around $p$ within a radius of $O(\sqrt{\log n})$ and there are $n$ points in $P$.
By setting $\alpha=O(\sqrt{\log n})$, we have 
\[
\text{$\abs{\mathcal{D}_{P,\sigma}(x)} < O(\sqrt{\log n})$ for \emph{all} $x\in\mathbb{R}$ with probability at least $1-1/10$}
\]
\emph{if} we manage to construct a coloring $\sigma$ satisfying \eqref{eq:property}.
Phillips and Tai \cite{phillips2019near} managed to construct such coloring $\sigma$ by Banaszczyk's Theorem and proved their result.
Namely, a coloring satisfying \eqref{eq:property} is construct-able.

\subparagraph{Attempts to improve the result}
We have seen how to show $\abs{\mathcal{D}_{P,\sigma}(x)} < O(\sqrt{\log n})$.
There is still a gap from showing $\abs{\mathcal{D}_{P,\sigma}(x)}<O(1)$.
We observe that the above argument aims at minimizing $\alpha$ such that the total failure probability $Ne^{-\Omega(\alpha^2)}$ is bounded by a constant.
If we manage to make the factor $N$ smaller, it helps setting $\alpha$ smaller and hence we can improve the result.

Recall that $N=O(n^2\sqrt{\log n}) = O(n\cdot n\sqrt{\log n})$ and the first factor $n$ comes from the fact that $P$ has $n$ points and these $n$ points could be widely spread out.
Namely, we need at most $n$ neighborhoods to cover all relevant grid points.
What if all points in $P$ lie inside a bounded region say $[-1,1]$? 
In this case, we just need to consider \emph{one} neighborhood to cover all relevant grid points.
Nonetheless, we do not assume that they are in a bounded region and we take care of it in the following way.
We partition $\mathbb{R}$ into infinitely many bounded regions (say $\dots, [-3,-1],[-1,1],[1,3],\dots$) and assign each point in $P$ to its corresponding region.
Then, we construct a coloring on the points in each bounded region and each coloring is constructed independently.
By triangle inequality, we have
\[
\abs{\mathcal{D}_{P,\sigma}(x)} \leq \sum \abs{\mathcal{D}_{P_i,\sigma_i}(x)} \numberthis \label{eq:breakd}
\]
where each $P_i\subset P$ is the set of points in the same bounded region and $\sigma_i$ is the coloring $\sigma$ restricted on $P_i$.

If we manage to construct the colorings $\sigma_i$ satisfying \eqref{eq:property} then we will end up getting $\abs{\mathcal{D}_{P,\sigma}(x)} < n_0 \cdot O(\alpha)$ where $n_0$ is the number of bounded regions that contain at least one point in $P$.
However, $n_0$ can be as large as $O(n)$.
To address this issue, we take the advantage of the assumption that all points in $P_i$ are in a bounded region (say $[-1,1]$).
Since all points in $P_i$ are in $[-1,1]$ now and Gaussians decay exponentially, intuitively we should be able to construct a coloring $\sigma_i$ that
\[
\text{$\abs{\mathcal{D}_{P_i,\sigma_i}(x)} < \alpha e^{-\frac{2}{3}x^2}$ for \emph{any} $x\in\mathbb{R}$ with probability at least $1-O(e^{-\Omega(\alpha^2)})$ for any $\alpha$}
\]
if a coloring satisfying \eqref{eq:property} is construct-able.
It is because we can rewrite $\abs{\mathcal{D}_{P_i,\sigma_i}(x)}$ as
\[
\abs{\mathcal{D}_{P_i,\sigma_i}(x)} =  \abs{\sum_{p\in P_i}\sigma_i(p)e^{-(x-p)^2}} = e^{-\frac{2}{3}x^2}\cdot\abs{\sum_{p\in P_i}\sigma_i(p)e^{2p^2}e^{-(\frac{1}{\sqrt{3}}x-\sqrt{3}p)^2}} \numberthis \label{eq:pullout}
\]
and the expression in the RHS has a form similar to $\mathcal{D}_{P_i,\sigma_i}(x)$.
The constant $\frac{2}{3}$ in the factor $e^{-\frac{2}{3}x^2}$ can be any constant between $0$ and $1$.
The extra factor $e^{-\frac{2}{3}x^2}$ is crucial: when we plug the bound $\alpha e^{-\frac{2}{3}x^2}$ into \eqref{eq:breakd}, $\abs{\mathcal{D}_{P,\sigma}(x)} $ is bounded by $O(1)\cdot O(\alpha)$ instead of $n_0\cdot O(\alpha)$.

One minor issue here is that the failure probability is accumulated when we ensure all $\sigma_i$ have the desired discrepancy.
We fix this issue by turning the construction of each $\sigma_i$ into a Las Vegas Algorithm.
Namely, we check if each $\sigma_i$ satisfies the desired discrepancy and repeat the construction if not.

Now, \emph{if} we manage to construct a coloring $\sigma$ such that: given $P\subset[-1,1]$, 
\[
\text{$\abs{\mathcal{D}_{P,\sigma}(x)} < \alpha e^{-\frac{2}{3}x^2}$ for \emph{any} $x\in \mathbb{R}$ with probability at least $1-O(e^{-\Omega(\alpha^2)})$ for any $\alpha$} \numberthis\label{eq:newproperty}
\]
then we only need to consider \emph{one} neighborhood to cover all relevant grid points when applying the union bound.
We also preview here that \eqref{eq:newproperty} is the only property a coloring needs to show our result.
From now on, we assume $P\subset[-1,1]$.
Even though \eqref{eq:newproperty} (the properties of the coloring $\sigma$ we are looking for) is slightly different than \eqref{eq:property} (what we stated in the beginning) because of the extra factor $e^{-\frac{2}{3}x^2}$, we can still perform a similar argument to prove that 
\[
\text{$\abs{\mathcal{D}_{P,\sigma}(x)} < O(\sqrt{\log n}) e^{-\frac{2}{3}x^2}$ for \emph{all} $x\in \mathbb{R}$ with probability at least $1-1/10$} \numberthis \label{eq:base}
\]
by arguing the slope of $\mathcal{D}_{P,\sigma}(x)$ is bounded by $O(n)e^{-\frac{2}{3}x^2}$ for any $x\in\mathbb{R}$.

\subparagraph{Reusing the guarantees for $\mathcal{D}_{P,\sigma}$}
Now, we look at the second factor $n\sqrt{\log n}$ in $N$.
It turns out that we are not going to make this factor smaller.
Instead, we will look at what guarantees this factor can give us and reuse these guarantees.

We further split $n\sqrt{\log n}$ into two parts: $n$ and $\sqrt{\log n}$.
Recall that the first part $n$ comes from the configuration that the cell width of the grid is $\Omega(1/n)$ and the second part $\sqrt{\log n}$ comes from the configuration that we need to consider the neighborhood of radius $O(\sqrt{\log n})$ to cover all relevant grid points.
However, we set up these two configurations \emph{without taking $\sigma$ into consideration}.
As we mentioned before, if we have a coloring $\sigma$ satisfying \eqref{eq:newproperty} then we have \eqref{eq:base}.
Can we reuse this guarantee and exploit the coloring $\sigma$?
To answer this question, we first investigate the term $\abs{\mathcal{D}_{P,\sigma}(x)-\mathcal{D}_{P,\sigma}(y)}$ for any $x,y\in\mathbb{R}$ and, by exploiting the structure of the Gaussians, we can prove
\[
\abs{\frac{\mathcal{D}_{P,\sigma}(x)-\mathcal{D}_{P,\sigma}(y)}{x-y}} < O(\abs{\xi})\cdot\abs{\mathcal{D}_{P,\sigma}(\xi)} \numberthis \label{eq:slope}
\]
for any $x\neq y$ where $\xi$ is in between $x$ and $y$.
\emph{The takeaway from this inequality is the slope of $\mathcal{D}_{P,\sigma}$ is bounded by $\mathcal{D}_{P,\sigma}$ itself.}
It is how we can reuse our guarantees.

If we plug our guarantee \eqref{eq:base} into \eqref{eq:slope}, we can show that the slope of $\mathcal{D}_{P,\sigma}(x)$ for this $\sigma$ is bounded by $O(\sqrt{\log n \log\log n}) e^{-\frac{2}{3}x^2}$ for any $x$ within a radius of $O(\sqrt{\log \log n})$.
For $x$ that lies beyond a radius of $\Omega(\sqrt{\log \log n})$, we have 
\[
\abs{\mathcal{D}_{P,\sigma}(x)} < O(\sqrt{\log n}) e^{-\frac{2}{3}x^2} <\frac{O(\sqrt{\log n})}{\Omega(\sqrt{\log n})} e^{-\frac{1}{3}x^2}<O(1) e^{-\frac{1}{3}x^2} < O(\sqrt{\log\log n}) e^{-\frac{1}{3}x^2} \numberthis \label{eq:outside}
\]
Note that the constant in the exponent becomes $\frac{1}{3}$ and it can be any constant smaller than $\frac{2}{3}$.
If we have a coloring $\sigma$ satisfying \emph{additionally} that $\abs{\mathcal{D}_{P,\sigma}(x)} < O(\sqrt{\log\log n}) e^{-\frac{2}{3}x^2}$ for \emph{all} $x$ in the set of grid points that are in the grid of cell width $\Omega(1/\sqrt{\log n})$ (instead of $\Omega(1/n)$) and bounded within a radius of $O(\sqrt{\log \log n})$ (instead of $O(\sqrt{\log n})$), then we have 
\[
\text{$\abs{\mathcal{D}_{P,\sigma}(x)} < O(\sqrt{\log\log n}) e^{-\frac{1}{3}x^2}$ for \emph{all} $x\in\mathbb{R}$.}
\]

There is a caveat: to ensure the coloring $\sigma$ satisfies the additional properties, we have to include more events in the union bound when invoking \eqref{eq:newproperty}.
In other words, the failure probability is now larger than $1/10$.
Nonetheless, we improved the previous result to $\abs{\mathcal{D}_{P,\sigma}(x)} < O(\sqrt{\log \log n}) e^{-\frac{1}{3}x^2}$.

\subparagraph{Hints of using induction}
From the improvement we just made, it gives us a hint to refine the quality of our result by induction.
One may notice the following pattern.
Suppose we have a coloring $\sigma$ satisfying
\[
\text{$\abs{\mathcal{D}_{P,\sigma}(x)} < \beta e^{-\kappa x^2}$ for \emph{all} $x\in \mathbb{R}$}\numberthis \label{eq:previousstep}
\]
for some $\beta$ where $\kappa$ is any constant between $0$ and $1$ (like $2/3$ before).
Let $S$ be the set of grid points that are in the grid of cell width $\Omega(1/\beta)$ and lie within a radius of $O(\sqrt{\log \beta})$.
Note that $\abs{S}=O(\beta\sqrt{\log \beta})$.
If this coloring $\sigma$ also satisfies that 
\[
\text{$\abs{\mathcal{D}_{P,\sigma}(x)} < O(\sqrt{\log \beta}) e^{-\kappa x^2}$ for \emph{all} $x\in S$} \numberthis \label{eq:inductivestepgrid}
\]
then we can modify the previous argument in the following way.
From \eqref{eq:previousstep} and \eqref{eq:slope}, we have the absolute value of the slope of $\mathcal{D}_{P,\sigma}$ at any point within a radius of $O(\sqrt{\log \beta})$ is bounded by $O(\beta\sqrt{\log \beta})e^{-\kappa x^2}$.
From an argument similar to \eqref{eq:outside}, we also have $\abs{\mathcal{D}_{P,\sigma}(x)} < O(\sqrt{\log \beta}) e^{-\kappa' x^2}$ for all $x$ that lies beyond a radius of $\Omega(\sqrt{\log \beta})$ where $\kappa'$ is any constant between $0$ and $\kappa$ (like $1/3$ before).
We combine them with \eqref{eq:inductivestepgrid} and it implies
\[
\text{$\abs{\mathcal{D}_{P,\sigma}(x)} < O(\sqrt{\log \beta}) e^{-\kappa'x^2}$ for \emph{all} $x\in \mathbb{R}$.} \numberthis \label{eq:inductivestep}
\]
If we take \eqref{eq:base} as the base step and the implication from \eqref{eq:previousstep} to \eqref{eq:inductivestep} as the inductive step, we should expect
\[
\text{$\abs{\mathcal{D}_{P,\sigma}(x)} < O(1) e^{-\frac{1}{3}x^2}$ for \emph{all} $x\in \mathbb{R}$.}
\]
after $O(\log^* n)$ inductive steps.

As we mentioned before, we also need to keep track of the failure probability and the exponent $\kappa$ in the factor $e^{-\kappa x^2}$.
We first deal with the failure probability.
In each inductive step, we need extra guarantees on the set of grid points of a smaller size (i.e. \eqref{eq:inductivestepgrid} when invoking \eqref{eq:newproperty}).
Hence, the total failure probability is a sum of $O(\log^* n)$ failure probabilities in each inductive step.
We can set these $O(\log^* n)$ failure probabilities to be a geometric sequence such that the total failure probability is a constant.
The other issue is the exponent.
We can again make this exponent decrease from $2/3$ to $1/3$ geometrically as it proceeds in the inductive steps.
In each inductive step, we need to set $\alpha$ in \eqref{eq:newproperty} larger than what we stated earlier accordingly when invoking \eqref{eq:newproperty} in the union bound.
Nonetheless, we eventually prove that $\abs{\mathcal{D}_{P,\sigma}(x)} < O(1)e^{-\frac{1}{3}x^2}$ for \emph{all} $x\in\mathbb{R}$ with probability $1/2$.

\subparagraph{Construction of the coloring}
It all boils down to the problem of how to construct a coloring $\sigma$ satisfying \eqref{eq:newproperty}.
Namely, given a point set $P\subset[-1,1]$,
\[
\text{$\abs{\mathcal{D}_{P,\sigma}(x)} < \alpha e^{-\frac{2}{3}x^2}$ for \emph{any} $x\in \mathbb{R}$ with probability at least $1-O(e^{-\Omega(\alpha^2)})$ for any $\alpha$.}
\]
We introduced Banaszczyk's Theorem (Theorem \ref{thm:b_thm}) before and if we can rewrite \eqref{eq:newproperty} as the inner product form shown in Theorem \ref{thm:b_thm} then we can apply the algorithm in Theorem \ref{thm:b_thm}.
As we mentioned in \eqref{eq:pullout}, we first rewrite
\[
\abs{\mathcal{D}_{P,\sigma}(x)} = \abs{\sum_{p\in P}\sigma(p)e^{-(x-p)^2}} = e^{-\frac{2}{3}x^2}\cdot\abs{\sum_{p\in P}\sigma(p)e^{2p^2}e^{-(\frac{1}{\sqrt{3}}x-\sqrt{3}p)^2}}.
\]
and hence we can ease the notation by dropping the factor $e^{-\frac{2}{3}x^2}$.
Namely, we need a coloring $\sigma$ such that, given a point set $P\subset[-1,1]$,
\[
\text{$\abs{\sum_{p\in P}\sigma(p)e^{2p^2}e^{-(\frac{1}{\sqrt{3}}x-\sqrt{3}p)^2}} < \alpha$ for \emph{any} $x\in \mathbb{R}$}
\]
with probability at least $1-O(e^{-\Omega(\alpha^2)})$ for any $\alpha$.
Since the Gaussian kernel is a positive-definite kernel, it implies that the term $e^{-(\frac{1}{\sqrt{3}}x-\sqrt{3}p)^2}$ can be rewritten as $\inner{u^{(\frac{1}{\sqrt{3}}x)}}{u^{(\sqrt{3}p)}}$ where $u^{(\cdot)}$ is a vector such that $\inner{u^{(s)}}{u^{(t)}} = e^{-(s-t)^2}$ for any $s,t\in \mathbb{R}$.
It is worth noting that $\norm{u^{(s)}}^2 = \inner{u^{(s)}}{u^{(s)}} = e^{-(s-s)^2} = 1$ for any $s\in\mathbb{R}$.
Hence, we further rewrite \eqref{eq:newproperty} as: given a point set $P\subset[-1,1]$,
\[
\text{$\abs{\inner{u^{(\frac{1}{\sqrt{3}}x)}}{\Sigma}} < \alpha$ for \emph{any} $x\in \mathbb{R}$ with probability at least $1-O(e^{-\Omega(\alpha^2)})$ for any $\alpha$}
\]
where $\Sigma = \sum_{p\in P} \sigma(p)e^{2p^2}u^{(\sqrt{3}p)}$.
It is the inner product form we are looking for in order to apply the algorithm in Theorem \ref{thm:b_thm}.
Recall that the norms of the input vectors and the query vectors in Banaszczyk's Theorem are required to be not larger than $1$.
We check that the norm of the query vector $\norm{u^{(\frac{1}{\sqrt{3}}x)}} = 1$ and the norm of the input vector $\norm{e^{2p^2}u^{(\sqrt{3}p)}} = O(1)$ since we assume that $P\subset[-1,1]$.
Karnin and Liberty \cite{karnin2019discrepancy} assumed \emph{both} the point set $P$ and the query $x$ lie within a constant radius because their result stops short of handling the norms of these vectors when using Banaszczyk's Theorem.
If we take $\frac{e^{2p^2}u^{(\sqrt{3}p)}}{\norm{e^{2p^2}u^{(\sqrt{3}p)}}}$ as the input vectors, we can apply the algorithm in Theorem \ref{thm:b_thm} to construct the desired coloring.

\section{Proofs}

In this section, we will show how to construct an $\eps$-coreset via discrepancy theory.
From now on, we assume that $d$ is a constant.
The $\log $ function in this paper is base $e$.
Also, we define the following notations.
Let $\Gr_d(\gamma)\subset\mathbb{R}^d$ be an infinite lattice grid of cell width $\gamma$, i.e. $\setdef{(\gamma i_1,\dots,\gamma i_d)}{i_1,\dots,i_d \text{ are integers}}$.
Denote $B_\infty^d(r) = \setdef{x}{\abs{x_j} < r\text{ for $j=1,\dots,d$}}$ to be a $\ell_\infty$-ball of radius $r$.
We define a decreasing sequence $n_i$ in the following way: $n_0 = \log^2 n$, $n_1 = \sqrt{3\log n}+3$ and $n_{i+1} = \sqrt{3\cdot 2^{\ell(n)-i}\log n_i}$ for $i=1,\dots,\ell(n) - 1$.
Here, $\ell(n)+3$ is the smallest integer $k$ that $\ilog(k, n)<0$ where $\ilog(k,n) = \log\cdots\log n$ (there are $k$ log functions) and it is easy to see that $\ell(n) = O(\log^* n)$.
For $i=0,\dots,\ell(n)-1$, denote $S_i = \Gr_d(\frac{1}{C_0n_i}) \cap [-n_{i+1},n_{i+1}]^d = \Gr_d(\frac{1}{C_0n_i}) \cap B_\infty^d(n_{i+1})$ where $C_0$ is a sufficiently large constant.
Namely, $S_i$ is a bounded lattice grid and its size is at most $(2C_0n_in_{i+1})^d$.
Note that $S_i$ may be interpreted as a subset of $S_0$ but, for clarity, we still view them as different sets.
Throughout this section, the absolute constants $C_0,C_1,C$ are unchanged.
Also, $C$ is larger than $C_1$ and $C_1$ is larger than $C_0$.

\subsection{Useful lemmas}\label{sec:useful}

Before we go into the main proof, we first present some important observations.
We delay all detailed proofs to the appendix.

\begin{restatable}{lemma}{lembaselipschitz}\label{lem:base_lipschitz}
	Suppose $P\subset B_\infty^d(1)$ be a point set of size $n$ and $\sigma$ is a coloring on $P$.
	Then, we have
	\[
	\sup_{x\in B_\infty^d(\sqrt{3\log n} + 3)}\abs{\sum_{p\in P} \sigma(p)e^{2\norm{p}^2} e^{-\frac{1}{3}\norm{x-3p}^2}} \leq 4\cdot\sup_{s\in S_0}\abs{\sum_{p\in P} \sigma(p)e^{2\norm{p}^2} e^{-\frac{1}{3}\norm{s-3p}^2}} + 7
	\]
	where $S_0 = \Gr_d(w) \cap [-\sqrt{3\log n}-3,\sqrt{3\log n} + 3]^d = \Gr_d(w) \cap B_\infty^d(\sqrt{3\log n}+3)$ with $w = \frac{1}{C_0\log^2 n}$.
	Here, $\Gr_d(\gamma) = \setdef{(\gamma i_1,\dots,\gamma i_d)}{i_1,\dots,i_d \text{ are integers}} \subset\mathbb{R}^d$ is an infinite lattice grid.
\end{restatable}

The main technique in Lemma \ref{lem:base_lipschitz} is expanding the expression by Taylor expansion.
Then, by truncating the Taylor expansion with a finite number of terms, one can bound the derivatives of the expression by using Markov Brother's inequality (Theorem \ref{lem:mb_ineq}).
Since the width of the grid cell in $S_0$ depends on the number of terms in the Taylor expansion, we need to argue that a small number of terms suffices to bound the error.

\begin{restatable}{lemma}{lemlipschitz}\label{lem:lipschitz}
	Given a coloring $\sigma$.
	For any $x,s \in  \mathbb{R}^d$ such that $\abs{x_j}\leq \abs{s_j}$ for all $j = 1,2,\dots,d$, we have
	\begin{align*}
		\MoveEqLeft \abs{\sum_{p\in P} \sigma(p) e^{-\norm{x-p}^2} - \sum_{p\in P} \sigma(p) e^{-\norm{s-p}^2} } \\
		& \leq 
		(\norm{s}^2-\norm{x}^2 )\abs{\sum_{p\in P} \sigma(p) e^{-\norm{x-p}^2}} + 2\sum_{j=1}^d\abs{s_j-x_j}\cdot\abs{\sum_{p\in P} \sigma(p) e^{-\norm{\xi^{(j)}-p}^2} }
	\end{align*}
	where $\xi^{(j)} = (x_1,\dots,x_{j-1},\xi_j,s_{j+1},\dots,s_d)$ for some $\xi_j$ between $\abs{x_j}$ and $\abs{s_j}$.
\end{restatable}

In the inductive step, the main observation is the absolute difference of the discrepancy objective at two different points, $\abs{\mathcal{D}_{P,\sigma}(x)-\mathcal{D}_{P,\sigma}(y)}$, can be bounded by the discrepancy objective itself.
Lemma \ref{lem:lipschitz} is the lemma providing the key inequality to perform the inductive steps.

Finally, we also show the asymptotic bound of the recurrence equation $n_i$ in Lemma \ref{lem:recurrence}.

\begin{restatable}{lemma}{lemrecurrence}\label{lem:recurrence}
	Let $n_0 = \log^2n$, $n_1 = \sqrt{3\log n} + 3$ and $n_{i+1} = \sqrt{3\cdot2^{\ell(n)-i}\log n_i}$ for $i =1,\dots, \ell(n)-1$.
	Then, $n_{\ell(n)} = O(1)$.
	Recall that $\ell(n)+3$ is the smallest integer $k$ that $\ilog(k,n) < 0$.
\end{restatable}

\subsection{Base step}\label{sec:base}

Recall that the definition of $\mathcal{D}_{P,\sigma}(x)$ is $\sum_{p\in P} \sigma(p)e^{-\norm{x-p}^2}$.
Lemma \ref{lem:base} shows that \emph{if} a coloring $\sigma$ satisfies that $\abs{\mathcal{D}_{P,\sigma}(x)}$ is small for all $x$ in a finite subset (which is a grid) of $\mathbb{R}^d$, then the coloring $\sigma$ also satisfies that $\abs{\mathcal{D}_{P,\sigma}(x)}$ is small for all $x\in\mathbb{R}^d$.
Note that we still haven't provided the detail on \emph{how} to find such coloring and we will do it in the full algorithm.

\begin{restatable}{lemma}{lembase}\label{lem:base}
	Suppose $P\subset B_\infty^d(1)$.
	Given a coloring $\sigma$ such that, for all $s'\in S_0 = \Gr_d(w) \cup B_\infty^d(n_1)$ where $w = \frac{1}{C_0\log^2 n} = \frac{1}{C_0n_0}$ is the same $w$ shown in Lemma \ref{lem:base_lipschitz} and $n_1 = \sqrt{3\log n} + 3$, 
	\[
	\abs{\mathcal{D}_{P,\sigma}(s')}=\abs{\sum_{p\in P} \sigma(p)e^{-\norm{s'-p}^2}} < C_1 n_1 e^{-\frac{2}{3}\norm{s'}^2}
	\]
	Then, we have, for all $x \in \mathbb{R}^d$,
	\[
	\abs{\mathcal{D}_{P,\sigma}(x)}=\abs{\sum_{p\in P}\sigma(p) e^{-\norm{x-p}^2}} < C n_1 e^{-\frac{2}{3}\norm{x}^2}.
	\]
	Here, $C,C_1,C_0$ are sufficiently large constant depending on $d$ only.
\end{restatable}

We make a short remark here.
One might notice that Lemma \ref{lem:base} states $w = \frac{1}{C_0\log^2 n}$ while it is sufficient to set $w = \Omega(\frac{1}{n})$ as suggested in Section \ref{sec:overview}.
As we mentioned before, our final algorithm is a Las Vegas algorithm and hence we need to check if the output coloring has the desired discrepancy.
We check it by enumerating the relevant grid points and computing the discrepancy at them.
Making $w$ larger reduces the size of the grid and hence improves the running time.
Nonetheless, $w = \frac{1}{C_0\log^2 n}=\Omega(\frac{1}{n})$ and hence it doesn't change the logic.

\subsection{Inductive step}\label{sec:inductive}

Lemma \ref{lem:inductive_lipschitz} suggests that \emph{if} a coloring $\sigma$ satisfies that $\abs{\mathcal{D}_{P,\sigma}(x)}$ is small for all $x\in\mathbb{R}^d$, then the coloring $\sigma$ satisfies that the absolute difference $\abs{\mathcal{D}_{P,\sigma}(x)-\mathcal{D}_{P,\sigma}(s)}$ is also small for any two close points $x,s$ within a certain region.
It is achieved by the observation that the slope of $\mathcal{D}_{P,\sigma}$ can be bounded by $\mathcal{D}_{P,\sigma}$ itself in magnitude.

\begin{restatable}{lemma}{leminductivelipschitz}\label{lem:inductive_lipschitz}
	Suppose $P \subset B_\infty^d(1)$.
	Let $D_i = C\cdot \frac{5}{4}(1-\frac{1}{5^i})$ and $I_i = \frac{1}{3}+\frac{1}{3}(1-\frac{1}{2^{\ell(n)-i}})$.
	Given a coloring $\sigma$ such that, for all $x\in\mathbb{R}^d$, 
	\[
	\abs{\mathcal{D}_{P,\sigma}(x)}=\abs{\sum_{p\in P} \sigma(p) e^{-\norm{x-p}^2}} < D_i\cdot n_ie^{-I_i\norm{x}^2}.
	\]
	If $x\in B_\infty^d(n_{i+1})$, then
	\[
	\abs{\mathcal{D}_{P,\sigma}(x)-\mathcal{D}_{P,\sigma}(s)}=\abs{\sum_{p\in P} \sigma(p) e^{-\norm{x-p}^2} - \sum_{p\in P} \sigma(p) e^{-\norm{s-p}^2} } \leq \frac{1}{5}D_i\cdot n_{i+1} e^{-I_i\norm{x}^2}.
	\]
	where $s\in S_i$ is the closest point to $x$ that $\abs{s_j}>\abs{x_j}$ for all $j=1,2,\dots,d$.
\end{restatable}

Similar to Lemma \ref{lem:base}, Lemma \ref{lem:inductive} shows that \emph{if} a coloring $\sigma$ satisfies that $\abs{\mathcal{D}_{P,\sigma}(x)}$ is small for all $x$ in a finite subset (which is a grid) of $\mathbb{R}^d$, then the coloring $\sigma$ also satisfies that $\abs{\mathcal{D}_{P,\sigma}(x)}$ is small for all $x\in\mathbb{R}^d$.
The only difference is that we can take the advantage of the discrepancy guarantee from the previous iterations.

\begin{restatable}{lemma}{leminductive}\label{lem:inductive}
	Suppose $P\subset B_\infty^d(1)$.
	Recall that $D_i = C\cdot \frac{5}{4}(1-\frac{1}{5^{i}})$ and $I_i = \frac{1}{3} + \frac{1}{3}(1-\frac{1}{2^{\ell(n)-i}})$ which is the same definition as in Lemma \ref{lem:inductive_lipschitz}.
	Given a coloring $\sigma$ such that, for all $s'\in S_i$,
	\[
	\abs{\sum_{p\in P} \sigma(p) e^{-\norm{s'-p}^2} } < C_1n_{i+1}e^{-\frac{2}{3}\norm{s'}^2}
	\]
	and, for all $x\in \mathbb{R}^d$,
	\[
	\abs{\sum_{p\in P} \sigma(p) e^{-\norm{x-p}^2}} \leq D_i\cdot n_{i} e^{-I_i\norm{x}^2}.
	\]
	Then, we have, for all $x\in\mathbb{R}^d$,
	\[
	\abs{\sum_{p\in P} \sigma(p) e^{-\norm{x-p}^2} } < D_{i+1}\cdot n_{i+1} e^{-I_{i+1}\norm{x}^2}.
	\]
	Here, $C, C_1$ are sufficiently large constants.
\end{restatable}

\subsection{Full algorithm}\label{sec:full}

For now, we still assume that $P \subset B_\infty^d(1)$.
Now, we can apply the algorithm in Theorem \ref{thm:b_thm} to construct our coloring $\sigma$ that produces a low discrepancy, $\abs{\mathcal{D}_{P,\sigma}(x)}$, for all $x\in\mathbb{R}^d$.
Recall that $\ell(n)+3$ is the smallest integer $k$ that $\ilog(k,n) < 0$.
Also, we defined $n_i$ before such that $n_0 = \log^2 n$, $n_1 = \sqrt{3\log n} + 3$ and $n_{i+1} = \sqrt{3\cdot 2^{\ell(n) - i} \log n_i}$.

\begin{restatable}{lemma}{lemboundeddisc} \label{lem:bounded_disc}
	Assuming $P \subset B_\infty^d(1)$.
	Given a set of vectors $V_P$ defined as follows.
	\[
	V_P=\setdef{\frac{1}{\sqrt{1+e^{4d}}} (
		\begin{matrix}
			1 \\
			v^{(p)}e^{2\norm{p}^2}
		\end{matrix}
		)}{p\in P}
	\]
	such that $\inner{v^{(p)}}{v^{(q)}} = e^{-3\norm{p-q}^2}$ for any $p,q\in P$.
	Then, by taking $V_P$ as the input, the algorithm in Theorem \ref{thm:b_thm} constructs a coloring $\sigma$ on $P$ such that 
	\[
	\abs{\sum_{p\in P} \sigma(p) e^{-\norm{x-p}^2} } < C\cdot \frac{5}{4}\cdot n_{\ell(n)} e^{-\frac{1}{3}\norm{x}^2}
	\]
	for all $x\in\mathbb{R}^d$ and $\abs{\sum_{p\in P} \sigma(p)} \leq C $ with probability at least $\frac{1}{2}$.  
\end{restatable}

Recall that we eventually would like to use the halving technique to construct our $\eps$-coreset.
To use the halving technique, we need to ensure that half of the points in $P$ are $+1$ and the other half are $-1$.
In Lemma \ref{lem:bounded_disc}, the $1$s concatenated on top of the vectors $v^{(p)}e^{2\norm{p}^2}$ in $V_P$ ensure the coloring has the above property.

\begin{restatable}{lemma}{lemboundeddiscbalance} \label{lem:bounded_disc_balance}
	Assuming $P \subset B_\infty^d(1)$.
	There is an efficient algorithm that constructs a coloring $\sigma$ such that $\abs{\sum_{p\in P} \sigma(p) e^{-\norm{x-p}^2} } = O(n_{\ell(n)} e^{-\frac{1}{3}\norm{x}^2})$ for all $x\in\mathbb{R}^d$ and half of points are assigned $+1$ with probability at least $\frac{1}{2}$.  
\end{restatable}

We can now remove the assumption of $P\subset B_\infty^d(1)$.
Algorithm \ref{alg:main_g} is a Las Vegas algorithm that constructs a coloring on the input point set $P$.
We can now show how to construct a coloring such that the discrepancy is small.
Recall that we defined $\Gr_d(\gamma) = \setdef{(\gamma i_1,\dots,\gamma i_d)}{i_1,\dots,i_d \text{ are integers}} \subset\mathbb{R}^d$ to be an infinite lattice grid.
The idea of Algorithm \ref{alg:main_g} is that we first decompose the entire $\mathbb{R}^d$ into infinitely many $\ell_\infty$-balls of radius $1$.
Then, we partition our input $P$ such that each point $p\in P$ lies in some $\ell_\infty$-ball.
For each non-empty $\ell_\infty$-ball, run the algorithm in Theorem \ref{thm:b_thm} to construct a coloring with the desired discrepancy by Lemma \ref{lem:bounded_disc_balance}.
Finally, we argue that there is an extra constant factor in the final discrepancy.

\begin{algorithm}[!h]
	\caption{Construction of the coloring}
	{\bf input}: a point set $P\subset \mathbb{R}^d$
	\begin{algorithmic}[1]
		\STATE initialize $Q_g=\emptyset$ for all $g\in \Gr_d(2)$
		\FOR {each $p\in P$}
		\STATE insert $p$ into $Q_g$ where $g\in\Gr_d(2)$ is the closest point to $p$.
		\ENDFOR
		\FOR {each non-empty $Q_g$}
		\STATE construct a collection $V_g$ of vector $\setdef{\frac{1}{\sqrt{1+e^{4d}}} (
			\begin{matrix}
				1 \\
				v^{(p)}e^{2\norm{p}^2}
			\end{matrix}
			)}{p\in Q_g}$ such that $\inner{v^{(p)}}{v^{(q)}} = e^{-3\norm{p-q}^2}$ for any $p,q\in Q_g$ \label{line:construct_input}
		\STATE use $V_g$ as the input and run the algorithm in Theorem \ref{thm:b_thm} to obtain a coloring $\sigma_g$ on $Q_g$ \label{line:main_step}
		\STATE check if $\sigma_g$ satisfies the conditions in Lemma \ref{lem:base} and Lemma \ref{lem:inductive} and repeat line \ref{line:main_step} if not \label{line:check}
		\STATE flip the color of any points such that half of points in $Q_g$ are colored $+1$.\label{line:flip}
		\ENDFOR
		\STATE \textbf{return} a coloring $\sigma: P  \rightarrow \{-1,+1\}$ such that $\sigma(p) = \sigma_g(p)$ when $p\in Q_g$
	\end{algorithmic}
	\label{alg:main_g}
\end{algorithm}

\begin{restatable}{lemma}{lemfull}
	Suppose $P\subset \mathbb{R}^d$ be a point set of size $n$.
	Then, Algorithm \ref{alg:main_g} constructs a coloring $\sigma$ on $P$ efficiently such that  $\sup_{x\in\mathbb{R}^d}\abs{\sum_{p\in P} \sigma(p)e^{-\norm{x-p}^2}} = O (1 )$ and half of the points in $P$ are colored $+1$.
\end{restatable}

One can first perform random sampling \cite{lopez2015towards} before running Algorithm \ref{alg:main_g} such that the input size $n = O(\frac{1}{\eps^2})$.
Finally, by the standard halving technique, we have the following theorem.

\begin{theorem}[Restated Theorem \ref{thm:main_g}]
	Suppose $P\subset \mathbb{R}^d$ be a point set of size $n$.
	Let $\Gb_P$ be the Gaussian kernel density estimate of $P$, i.e. $\Gb_P(x) = \frac{1}{\abs{P}}\sum_{p\in P} e^{-\norm{x-p}^2}$ for any $x\in\mathbb{R}^d$.
	For a fixed constant $d$, there is an algorithm that constructs a subset $Q\subset P$ of size $O (\frac{1}{\eps} )$ such that $\sup_{x\in\mathbb{R}^d}\abs{\Gb_P(x) - \Gb_Q(x)} < \eps$ and has a polynomial running time in $n$.
\end{theorem}

\section{Conclusion and discussion}

In this paper, we studied the question of constructing coresets for kernel density estimates.
We proved that the Gaussian kernel has an $\eps$-coreset of the optimal size $O (1/\eps )$ when $d$ is a constant.
This coreset can be constructed efficiently.
We leveraged Banaszczyk's Theorem to construct a coloring such that the kernel discrepancy is small.
Then, we constructed an $\eps$-coreset of the desired size via the halving technique.

Some open problems in discrepancy theory, such as Tusn{\'a}dy's Problem, have an issue that an extra factor shows up when we generalize the result from the case of $d=1$ to the case of larger $d$.
A previous result of our problem is reducing our problem to Tusn{\'a}dy's Problem.
It turns out that, if $d=1$, the trivial solution gives the optimal result.
Unfortunately, it cannot be generalized to the higher dimensional case.
Our new induction analysis combining with Banaszczyk's Theorem provides a non-trivial perspective even when $d=1$.
Hence, it might open up a possibility of improving the results on these open problems.

Even though the Gaussian kernel is a major class of kernels in most of applications, it would be interesting to investigate similar results on other kernel settings such as the Laplace kernel.
Our approach exploits the properties of the Gaussian kernel such as factoring out the $e^{-\Omega(\norm{x}^2)}$ factor while maintaining the positive-definiteness.
Generalizing the result to a broader class of kernels might require deeper understandings of the properties that other kernels share with the Gaussian kernel.

In some applications, the input data might be in the high dimensional space.
Our result assumes that $d$ is a constant.
Note that one of the previous results is sub-optimal in terms of $\eps$ but is optimal in terms of $d$.
The dependence on $d$ in our result is exponential which we might want to avoid in the high dimensional case.
Hence, improving the dependence on $d$ to polynomial is also interesting because it would be more practical in some applications.

%%
%% Bibliography
%%

%% Please use bibtex, 
\bibliographystyle{plain}
\bibliography{ref}

\appendix

\section{Omitted proofs}

\subsection{Proofs in Section \ref{sec:useful}}

\lembaselipschitz*

\begin{proof}
	
	We first use Taylor expansion to have the following expression.
	\begin{align*}
		\MoveEqLeft\sum_{p\in P} \sigma(p)e^{2\norm{p}^2} e^{-\frac{1}{3}\norm{x-3p}^2} \\
		& =
		\sum_{p\in P} \sigma(p)e^{-\norm{p}^2}e^{-\frac{1}{3}\norm{x}^2} e^{2\inner{x}{p}}\\
		& =
		\sum_{p\in P} \sigma(p)e^{-\norm{p}^2}e^{-\frac{1}{3}\norm{x}^2}\sum_{k=0}^\infty \frac{(2\inner{x}{p})^k}{k!} \\
		& =
		\sum_{k=0}^\infty\inner{\sum_{p\in P} \sigma(p)e^{-\norm{p}^2}\sqrt{\frac{2^k}{k!}}p^{\otimes k} }{e^{-\frac{1}{3}\norm{x}^2}\sqrt{\frac{2^k}{k!}}x^{\otimes k}} \numberthis\label{eqn:disc_taylor}
	\end{align*}
	Here, $\otimes$ is Kronecker product.
	Namely, for any two vectors $a^{(1)}\in \mathbb{R}^{d_1}$ and $a^{(2)}\in\mathbb{R}^{d_2}$, $a^{(1)}\otimes a^{(2)}$ is a $d_1\cdot d_2$ dimensional vector indexed by two integers $i,j$ for $i=1,\dots,d_1$ and $j=1,\dots,d_2$ such that $(a^{(1)}\otimes a^{(2)})_{i,j} = a^{(1)}_i \cdot a^{(2)}_j$.
	Also, for any $a\in \mathbb{R}^d$ and any integer $k$, $a^{\otimes k} = a\otimes a^{\otimes k-1}$ and $a^{\otimes 0} = 1$.
	
	Now, consider polynomial $\mathcal{P}_{P,\rho}(x) = \sum_{k=0}^\rho\inner{\sum_{p\in P} \sigma(p)e^{-\norm{p}^2}\sqrt{\frac{2^k}{k!}}p^{\otimes k} }{\sqrt{\frac{2^k}{k!}}x^{\otimes k}}$.
	We first observe that $e^{-\frac{1}{3}\norm{x}^2}\mathcal{P}_{P,\rho}(x)$ is the same as expression (\ref{eqn:disc_taylor}) but is truncated at the $\rho$-th term.
	We analyze the term $\abs{\sum_{p\in P} \sigma(p)e^{2\norm{p}^2} e^{-\frac{1}{3}\norm{x-3p}^2} - e^{-\frac{1}{3}\norm{x}^2}\mathcal{P}_{P,\rho}(x)}$ for any $x\in B_\infty^d(\sqrt{3\log n}+3)$.
	\allowdisplaybreaks
	\begin{align*}
		\MoveEqLeft \abs{\sum_{p\in P} \sigma(p)e^{2\norm{p}^2} e^{-\frac{1}{3}\norm{x-3p}^2} - e^{-\frac{1}{3}\norm{x}^2}\mathcal{P}_{P,\rho}(x)}\\ 
		& =
		\abs{\sum_{k=\rho+1}^\infty\inner{\sum_{p\in P} \sigma(p)e^{-\norm{p}^2}\sqrt{\frac{2^k}{k!}}p^{\otimes k} }{e^{-\frac{1}{3}\norm{x}^2}\sqrt{\frac{2^k}{k!}}x^{\otimes k}}} \\
		& \leq
		\sum_{k=\rho+1}^\infty\frac{2^k}{k!}\sum_{p\in P} e^{-\norm{p}^2}e^{-\frac{1}{3}\norm{x}^2}\abs{\inner{p^{\otimes k} }{x^{\otimes k}}} \\
		& \leq
		\sum_{k=\rho+1}^\infty\frac{2^k}{k!}\sum_{p\in P} e^{-\norm{p}^2}e^{-\frac{1}{3}\norm{x}^2}\norm{p}^{k} \norm{x}^{ k} \\
		& =
		\sum_{p\in P}e^{-\norm{p}^2}e^{-\frac{1}{3}\norm{x}^2}\sum_{k=\rho+1}^\infty\frac{(2\norm{x}\norm{p})^k}{k!}
	\end{align*}
	By the error approximation of Taylor expansion of exponential function, we have $\sum_{k=\rho+1}^\infty\frac{y^k}{k!} \leq  (\sup_{\xi\in [-y, y]}e^\xi )\frac{y^{\rho+1}}{(\rho+1)!} = e^{y}\frac{y^{\rho+1}}{(\rho+1)!}$ where $y = 2\norm{x}\norm{p}$.
	Then,
	\begin{align*}
		\MoveEqLeft \abs{\sum_{p\in P} \sigma(p)e^{2\norm{p}^2} e^{-\frac{1}{3}\norm{x-3p}^2} - e^{-\frac{1}{3}\norm{x}^2}\mathcal{P}_{P,\rho}(x)} \\
		& \leq
		\sum_{p\in P}e^{-\norm{p}^2}e^{-\frac{1}{3}\norm{x}^2}e^{2\norm{x}\norm{p}}\frac{(2\norm{x}\norm{p})^{\rho+1}}{(\rho+1)!}\\
		& = 
		\sum_{p\in P}e^{2\norm{p}^2-\frac{1}{3}(3\norm{p}-\norm{x})^2}\frac{(2\norm{x}\norm{p})^{\rho+1}}{(\rho+1)!} \\
		& \leq
		\sum_{p\in P} e^{2d} \frac{(2\cdot \sqrt{d}\cdot \sqrt{d}(\sqrt{3\log n}+3))^{\rho+1}}{(\rho+1)!} \\
		& \leq
		\sum_{p\in P} e^{2d}  (\frac{2ed(\sqrt{3\log n}+3)}{\rho+1} )^{\rho+1}
	\end{align*}
	The last step is due to the fact of $z! \geq (\frac{z}{e})^{z}$.
	Now, if taking $\rho+1 = 2e^2d(\sqrt{3\log n}+3) + \log n + 2d$ which means $\rho = O(\log n)$, we have 
	\begin{align*}
		\MoveEqLeft \abs{\sum_{p\in P} \sigma(p)e^{2\norm{p}^2} e^{-\frac{1}{3}\norm{x-3p}^2} - e^{-\frac{1}{3}\norm{x}^2}\mathcal{P}_{P,\rho}(x)} \\
		& \leq
		\sum_{p\in P}e^{2d}(\frac{1}{e})^{\rho+1}
		\leq
		ne^{2d}(\frac{1}{e})^{\log n +2d}
		=
		1 \numberthis \label{eqn:truncate_error}
	\end{align*}
	
	Let $x^*$ be $\arg\sup_{x\in B_\infty^d(\sqrt{3\log n}+3)}\abs{\sum_{p\in P} \sigma(p)e^{2\norm{p}^2} e^{-\frac{1}{3}\norm{x-3p}^2}}$.
	Without loss of generality, we assume that all $x^*_i>0$ for all $i=1,\dots,d$.
	We construct an uniform grid $S_0$ of width $w=\frac{1}{C_0\log^2 n}<\frac{1}{8d\rho^2e^{d} + 2 d(\sqrt{3\log n}+3)/\log 2}$ on $B_\infty^d(\sqrt{3\log n}+3)$.
	Namely, $S_0 = \Gr_d(w) \cap B_\infty^d(\sqrt{3\log n}+3)$.
	Let $s$ be the closest grid point in $\mathsf{G}_w$ to $x^*$ that all coordinate of $s$ are larger than $x^*$.
	
	We have
	\[
	\abs{\mathcal{P}_{P,\rho}(x^*) - \mathcal{P}_{P,\rho}(s)}
	=
	\abs{\inner{\nabla\mathcal{P}_{P,\rho}(\xi)}{x^* - s}}
	\]
	for some $\xi$ on the segment between $x^*$ and $s$.
	Here, $\nabla$ is the gradient.
	Then, 
	\begin{align*}
		\abs{\mathcal{P}_{P,\rho}(x^*) - \mathcal{P}_{P,\rho}(s)}
		& =
		\abs{\inner{\nabla\mathcal{P}_{P,\rho}(\xi)}{x^* - s}} \\
		& \leq 
		\sum_{i=1}^d \abs{\frac{\partial\mathcal{P}_{P,\rho}}{\partial x_i}(\xi)}\abs{x^*_i - s_i} \\
		& \leq
		\sum_{i=1}^d 2\rho^2\abs{\mathcal{P}_{P,\rho}(x')}\abs{x^*_i - s_i}
	\end{align*}
	where $x' = \arg\max_{x\in\prod_{i=1}^d[s_i-1,s_i]}\mathcal{P}_{P,\rho}(x)$.
	The last step is due to Markov brother's inequality (Theorem \ref{lem:mb_ineq}).
	By rearranging the terms and multiplying $e^{-\norm{g}^2}$, we have
	\[
	e^{-\norm{s}^2}\abs{\mathcal{P}_{P,s}(x^*)} \leq         e^{-\norm{s}^2}\abs{\mathcal{P}_{P,s}(s)} +  (\sum_{i=1}^d 2s^2\abs{x^*_i - s_i} )e^{-\norm{s}^2}\abs{\mathcal{P}_{P,s}(x')}
	\]
	
	For the index $i$ that $s_i\geq \frac{1}{2}$, we have $x'_i\leq s_i$ and $-x'_i\leq -s_i+1 \leq \frac{1}{2}\leq s_i$, which means $\abs{x'_i} \leq \abs{s_i}$.
	For the index $i$ that $s_i < \frac{1}{2}$, recalling that we assume $x^*_i>0$ and so $s_i>0$, we have $\abs{x'_i} \leq 1-s_i \leq 1$. 
	Combining the above two conclusions, we have $\abs{x'_i}^2 - \abs{s_i}^2 \leq 1$ for all $i=1,\dots,d$.
	\begin{align*}
		e^{-\norm{s}^2}\abs{\mathcal{P}_{P,\rho}(x^*)} 
		&\leq 
		e^{-\norm{s}^2}\abs{\mathcal{P}_{P,\rho}(s)} +  (\sum_{i=1}^d 2\rho^2\abs{x^*_i - s_i} )e^{-\norm{s}^2}\abs{\mathcal{P}_{P,\rho}(x')} \\
		& \leq
		e^{-\norm{s}^2}\abs{\mathcal{P}_{P,\rho}(s)} +  (\sum_{i=1}^d 2\rho^2\abs{x^*_i - s_i} )e^{\norm{x'}^2-\norm{s}^2}e^{-\norm{x'}^2}\abs{\mathcal{P}_{P,\rho}(x')} \\
		& \leq
		e^{-\norm{s}^2}\abs{\mathcal{P}_{P,\rho}(s)} +  2\rho^2dwe^{d}e^{-\norm{x'}^2}\abs{\mathcal{P}_{P,s}(x')} \\
		& \leq
		e^{-\norm{s}^2}\abs{\mathcal{P}_{P,s}(s)} + \frac{1}{4}e^{-\norm{x'}^2}\abs{\mathcal{P}_{P,s}(x')} 
	\end{align*}
	In last step, recall that $w < \frac{1}{8d\rho^2e^{d} + 2 d(\sqrt{3\log n}+3)/\log 2} \leq \frac{1}{8d\rho^2e^{d}}$.
	The left hand side is basically $e^{-\norm{s}^2}\abs{\mathcal{P}_{P,\rho}(x^*)} = e^{-(\norm{s}^2-\norm{x^*}^2)}e^{-\norm{x^*}^2}\abs{\mathcal{P}_{P,\rho}(x^*)}$.
	We first analyze the term  $\norm{s}^2 - \norm{x^*}^2$.
	Recall that $w< \frac{1}{8d\rho^2e^{d} + 2 d(\sqrt{3\log n}+3)/\log 2} \leq \frac{1}{2 d(\sqrt{3\log n}+3)/\log 2}$.
	\begin{align*}
		\norm{s}^2 - \norm{x^*}^2
		& = 
		\inner{s-x^*}{s+x^*} \leq \sum_{i=1}^d\abs{s_i-x^*_i}\abs{s_i+x^*_i} \\
		& \leq 
		2(\sqrt{3\log n}+3)\sum_{i=1}^d\abs{s_i-x^*_i} \\
		& \leq 
		2(\sqrt{3\log n}+3)dw \\
		& \leq
		\log 2
	\end{align*}
	which means $e^{-\norm{s}^2}\abs{\mathcal{P}_{P,\rho}(x^*)} \geq \frac{1}{2}e^{-\norm{x^*}^2}\abs{\mathcal{P}_{P,\rho}(x^*)}$.
	Therefore, we have 
	\[
	\frac{1}{2}e^{-\norm{x^*}^2}\abs{\mathcal{P}_{P,\rho}(x^*)} \leq e^{-\norm{s}^2}\abs{\mathcal{P}_{P,\rho}(s)} + \frac{1}{4}e^{-\norm{x'}^2}\abs{\mathcal{P}_{P,\rho}(x')} 
	\]
	and then, from (\ref{eqn:truncate_error}), 
	\begin{align*}
		&
		\frac{1}{2} (\abs{\sum_{p\in P} \sigma(p)e^{2\norm{p}^2} e^{-\frac{1}{3}\norm{x^*-3p}^2}} - 1 ) \\
		& \leq 
		(\abs{\sum_{p\in P} \sigma(p)e^{2\norm{p}^2} e^{-\frac{1}{3}\norm{s-3p}^2}} + 1 ) + \frac{1}{4} (\abs{\sum_{p\in P} \sigma(p)e^{2\norm{p}^2} e^{-\frac{1}{3}\norm{x'-3p}^2}} + 1 ) \\
		& \leq
		(\abs{\sum_{p\in P} \sigma(p)e^{2\norm{p}^2} e^{-\frac{1}{3}\norm{s-3p}^2}} + 1 ) + \frac{1}{4} (\abs{\sum_{p\in P} \sigma(p)e^{2\norm{p}^2} e^{-\frac{1}{3}\norm{x^*-3p}^2}} + 1 )
	\end{align*}
	or
	\[
	\abs{\sum_{p\in P} \sigma(p)e^{2\norm{p}^2} e^{-\frac{1}{3}\norm{x^*-3p}^2}}  \leq 4\abs{\sum_{p\in P} \sigma(p)e^{2\norm{p}^2} e^{-\frac{1}{3}\norm{s-3p}^2}} + 7 \qedhere
	\]

\end{proof}

\lemlipschitz*

\begin{proof}
	
	We first express the term in the following form. 
	\begin{align*}
		\MoveEqLeft \abs{\sum_{p\in P} \sigma(p) e^{-\norm{x-p}^2} - \sum_{p\in P} \sigma(p) e^{-\norm{s-p}^2} } \\
		& = 
		\abs{\sum_{p\in P} \sigma(p)  (e^{-\norm{x-p}^2} - e^{-\norm{s}^2+2\inner{x}{p}-\norm{p}^2} + e^{-\norm{s}^2+2\inner{x}{p}-\norm{p}^2} - e^{-\norm{s-p}^2} ) } \\
		& \leq
		\abs{\sum_{p\in P} \sigma(p)  (e^{-\norm{x-p}^2} - e^{-\norm{s}^2+2\inner{x}{p}-\norm{p}^2}  )} \\
		& \qquad
		+ \abs{\sum_{p\in P} \sigma(p)  ( e^{-\norm{s}^2+2\inner{x}{p}-\norm{p}^2} - e^{-\norm{s-p}^2}  ) } \\
		& =
		(e^{-\norm{x}^2} - e^{-\norm{s}^2} )\abs{\sum_{p\in P} \sigma(p) e^{2\inner{x}{p}-\norm{p}^2}} \\
		& \qquad
		+ e^{-\norm{s}^2} \abs{\sum_{p\in P} \sigma(p)  ( e^{2\inner{x}{p}-\norm{p}^2} - e^{2\inner{s}{p}-\norm{p}^2}  ) }
	\end{align*}
	By mean value theorem, we have $e^{2\inner{x}{p}-\norm{p}^2} - e^{2\inner{s}{p}-\norm{p}^2} = \sum_{j=1}^d 2(x_j-s_j)e^{2\inner{\xi^{(j)}}{p}-\norm{p}^2}$ where $\xi^{(j)} = (x_1,\dots,x_{j-1},\xi_j,s_{j+1},\dots,s_d)$ for some $\xi_j$ between $\abs{x_j}$ and $\abs{s_j}$.
	Then, 
	\begin{align*}
		\MoveEqLeft \abs{\sum_{p\in P} \sigma(p) e^{-\norm{x-p}^2} - \sum_{p\in P} \sigma(p) e^{-\norm{s-p}^2} } \\
		& =
		(e^{-\norm{x}^2} - e^{-\norm{s}^2} )\abs{\sum_{p\in P} \sigma(p) e^{2\inner{x}{p}-\norm{p}^2}} \\
		& \qquad
		+ e^{-\norm{s}^2}\abs{\sum_{p\in P} \sigma(p)\sum_{j=1}^d 2(x_j-s_j)e^{2\inner{\xi^{(j)}}{p}-\norm{p}^2} } \\
		& \leq
		(1 - e^{\norm{x}^2-\norm{s}^2} )\abs{\sum_{p\in P} \sigma(p) e^{-\norm{x-p}^2}} \\
		& \qquad
		+ 2\sum_{j=1}^d\abs{s_j-x_j}e^{-\norm{s}^2+\norm{\xi^{(j)}}^2}\abs{\sum_{p\in P} \sigma(p) e^{-\norm{\xi^{(j)}-p}^2} } \\
		& \leq
		(\norm{s}^2-\norm{x}^2 )\abs{\sum_{p\in P} \sigma(p) e^{-\norm{x-p}^2}} + 2\sum_{j=1}^d\abs{s_j-x_j}\abs{\sum_{p\in P} \sigma(p) e^{-\norm{\xi^{(j)}-p}^2} }
	\end{align*}
	The last line follows from the fact of $1+y\leq e^y$ and $\norm{\xi^{(j)}} \leq \norm{s}$.
	
\end{proof}

\lemrecurrence*

\begin{proof}
	
	We first show that $2^{\ell(n)-i} < \ilog(i+1,n)$.
	Since $2^{\ell(n)-i} < e^{\ell(n)-i}$, it suffices to show that $\ell(n)-i < \ilog(i+2, n)$.
	Suppose it is not the case.
	Namely, $\ilog(i+2, n) > \ell(n)-i$.
	For any $x>0$, we have $1+\log x < e^{\log x} = x$ or, equivalently, $\log x < x-1$.
	By induction, it implies that, for any positive integer $t$, $\ilog(t,x) = \log(\ilog(t-1,x)) < \ilog(t-1,x)-1 < x-t$ as long as $\ilog(t-1,x) >0$.
	If we take $t = \ell(n)-i$ and $x = \ilog(i+2,n)$, then we have 
	\[
	\ilog(\ell(n)+2, n) = \ilog(\ell(n)-i,\ilog(i+2,n)) < \ilog(i+2,n) - (\ell(n) - i)<0.
	\]
	However, by definition, $\ell(n)+3$ is the smallest integer $k$ that $\ilog(k,n) < 0$ which means $\ilog(\ell(n)+2, n) > 0$.
	We can conclude that $2^{\ell(n)-i} < \ilog(i+1,n)$.
	
	Consider $n_1$, $n_1 = \sqrt{3\log n}+3 < \sqrt{3}\log n$ for large $n$.
	Suppose $n_{i} < \sqrt{6(1-\frac{1}{2^i})}\ilog(i,n)$ for $i  =1,\dots, \ell(n)-1$.
	We can show that 
	\begin{align*}
		n_{i+1} 
		& = 
		\sqrt{3\cdot 2^{\ell(n)-i}\log n_i} \\
		& \leq
		\sqrt{3\ilog(i+1,n)\log (\sqrt{6(1-\frac{1}{2^i})}\ilog(i,n) )} \\
		& < 
		\sqrt{3\ilog(i+1,n) (\frac{1}{2}\log(6(1-\frac{1}{2^i}))+\ilog(i+1,n) )} \\
		& <
		\sqrt{3\ilog(i+1,n) ((1-\frac{1}{2^i})\log 3+\ilog(i+1,n) )}\\
		& <
		\sqrt{3\ilog(i+1,n) ((1-\frac{1}{2^i})\ilog(i+1,n)+\ilog(i+1,n) )}\\
		& =
		\sqrt{6(1-\frac{1}{2^{i+1}})}\ilog(i+1,n)
	\end{align*}
	In the fourth line, we uses the fact that $\log(6(1-\frac{1}{2^i})) < 2(1-\frac{1}{2^i})\log 3$ and, in the fifth line, we uses the fact that $\log 3 < e^e< \ilog(\ell(n),n) \leq \ilog(i+1,n)$.
	In particular, we conclude 
	\[
	n_{\ell(n)} <\sqrt{6(1-\frac{1}{2^{\ell(n)}})}\ilog(\ell(n),n) < \sqrt{6}e^{e^e}
	\]
	
\end{proof}

\subsection{Proofs in Section \ref{sec:base}}

\lembase*

\begin{proof}
	
	We first have the following expression.
	\begin{align*}
		\abs{\sum_{p\in P} \sigma(p) e^{-\norm{x-p}^2}}
		& =
		\abs{\sum_{p\in P} \sigma(p) e^{-\frac{2}{3}\norm{x}^2-\frac{1}{3}\norm{x-3p}^2+2\norm{p}^2}} \\
		& =
		e^{-\frac{2}{3}\norm{x}^2}\abs{\sum_{p\in P} \sigma(p) e^{-\frac{1}{3}\norm{x-3p}^2}e^{2\norm{p}^2}} 
	\end{align*}
	
	If $x\notin B_\infty^d(\sqrt{3\log n}+3)$, which means one of the coordinate $\abs{x_j}>n_1 = \sqrt{3\log n}+3$, then
	\[
	\abs{\sum_{p\in P} \sigma(p)e^{-\norm{x-p}^2}} = e^{-\frac{2}{3}\norm{x}^2}\abs{\sum_{p\in P} \sigma(p) e^{-\frac{1}{3}\norm{x-3p}^2}e^{2\norm{p}^2}} < e^{-\frac{2}{3}\norm{x}^2}\sum_{p\in P} e^{-\frac{1}{3}\norm{x-3p}^2}e^{2\norm{p}^2}
	\]
	by triangle inequality.
	Since we assume that $\norm{p}<1$, the term $e^{2\norm{p}^2}$ is less than $e^{2d^2}$.
	Also, by the fact that $\norm{p}<1$ and $\abs{x_j} > \sqrt{3\log n}+3$, we have $\norm{x-3p} > \abs{x_j - 3p_j} > \abs{x_j} - 3\abs{p_j} > \sqrt{3\log n}$ which implies $e^{-\frac{1}{3}\norm{x-3p}^2} < \frac{1}{n}$.
	Now,
	\begin{align*}
		\abs{\sum_{p\in P} \sigma(p)e^{-\norm{x-p}^2}} 
		& < 
		e^{-\frac{2}{3}\norm{x}^2}\sum_{p\in P} e^{-\frac{1}{3}\norm{x-3p}^2}e^{2\norm{p}^2} 
		<
		e^{-\frac{2}{3}\norm{x}^2}\sum_{p\in P} \frac{1}{n}e^{2d^2} \\
		& =
		e^{2d^2}e^{-\frac{2}{3}\norm{x}^2} 
		%		 <
		%		C\sqrt{\log\log n}e^{-\frac{2}{3}\norm{x}^2}
		<
		Cn_1e^{-\frac{2}{3}\norm{x}^2}
	\end{align*}
	The last inequality follows from the assumption that $C$ is a sufficiently large constant (say at least $e^{2d^2}$).
	
	If $x\in B_\infty^d(\sqrt{3\log n}+3)$, by the assumption of $\sigma$, we rewrite the inequality to be , for all $s'\in S_0$, 
	\[
	\abs{\sum_{p\in P} \sigma(p)e^{2\norm{p}^2}e^{-\frac{1}{3}\norm{s'-3p}^2}} < C_1 n_1
	\]
	Now, by Lemma \ref{lem:base_lipschitz}, we have, for all $x\in \mathbb{R}^d$,
	\begin{align*}
		\abs{\sum_{p\in P} \sigma(p) e^{-\norm{x-p}^2}} 
		& =
		e^{-\frac{2}{3}\norm{x}^2}\abs{\sum_{p\in P} \sigma(p)e^{2\norm{p}^2} e^{-\frac{1}{3}\norm{x-3p}^2}} \\
		& \leq
		e^{-\frac{2}{3}\norm{x}^2} (4\sup_{s\in S_0}\abs{\sum_{p\in P} \sigma(p)e^{2\norm{p}^2} e^{-\frac{1}{3}\norm{s-3p}^2}} + 7 ) \\
		& \leq
		e^{-\frac{2}{3}\norm{x}^2} (4C_1n_1 + 7 ) \\
		& \leq
		C\cdot n_1e^{-\frac{2}{3}\norm{x}^2}
	\end{align*}
	The last inequality comes from the fact that $C$ is a sufficiently large constant (say at least $4C_1+7$).

\end{proof}

\subsection{Proofs in Section \ref{sec:inductive}}

\leminductivelipschitz*

\begin{proof}
	
	From Lemma \ref{lem:lipschitz}, we have
	\begin{align*}
		\MoveEqLeft \abs{\sum_{p\in P} \sigma(p) e^{-\norm{x-p}^2} - \sum_{p\in P} \sigma(p) e^{-\norm{s-p}^2} } \\
		& \leq 
		(\norm{s}^2-\norm{x}^2 )\abs{\sum_{p\in P} \sigma(p) e^{-\norm{x-p}^2}} + 2\sum_{j=1}^d\abs{s_j-x_j}\abs{\sum_{p\in P} \sigma(p) e^{-\norm{\xi^{(j)}-p}^2} }
	\end{align*}
	
	We first analyze the first term.
	By the assumption of $\sigma$,
	\[
	(\norm{s}^2-\norm{x}^2 )\abs{\sum_{p\in P} \sigma(p) e^{-\norm{x-p}^2}} 
	\leq
	(\sum_{j=1}^d\abs{s_j-x_j}\abs{s_j+x_j} )D_i\cdot n_{i} e^{-I_i\norm{x}^2}
	\]
	We recall that the width of $S_i$ is $\frac{1}{C_0n_i}$ and $x\in B_\infty^d(n_{i+1})$ which implies the term $\sum_{j=1}^d\abs{s_j-x_j}\abs{s_j+x_j}$ is less than  $2d\cdot\frac{1}{C_0n_i}\cdot n_{i+1}$.
	Therefore, we have
	\begin{align*}
		(\norm{s}^2-\norm{x}^2 )\abs{\sum_{p\in P} \sigma(p) e^{-\norm{x-p}^2}}
		& <
		(2d\cdot\frac{1}{C_0n_i}\cdot n_{i+1} )D_i\cdot n_{i} e^{-I_i\norm{x}^2} \\ 
		& \leq
		\frac{1}{10}D_i\cdot n_{i+1} e^{-I_i\norm{x}^2}
	\end{align*}
	The last line uses the assumption of sufficiently large constant $C_0$ (say at least $20d$).
	
	Now, we analyze the second term.
	By the assumption of $\sigma$, we have
	\[
	2\sum_{j=1}^d\abs{s_j-x_j}\abs{\sum_{p\in P} \sigma(p) e^{-\norm{\xi^{(j)}-p}^2} } 
	<
	2\sum_{j=1}^d\abs{s_j-x_j}D_i\cdot n_{i} e^{-I_i\norm{x}^2}
	\]
	Recall that the width of the grid $S_i$ is $\frac{1}{C_0n_i}$ which implies the term $\sum_{j=1}^d\abs{s_j-x_j}$ is less than $d\cdot\frac{1}{C_0n_i}$.
	\begin{align*}
		2\sum_{j=1}^d\abs{s_j-x_j}\abs{\sum_{p\in P} \sigma(p) e^{-\norm{\xi^{(j)}-p}^2} } 
		& <
		2d\frac{1}{C_0n_i}D_i\cdot n_{i} e^{-I_i\norm{x}^2} \\
		& \leq
		\frac{1}{10}D_i\cdot e^{-I_i\norm{x}^2} 
		\leq
		\frac{1}{10}D_i\cdot n_{i+1} e^{-I_i\norm{x}^2}
	\end{align*}
	The last line uses the assumption of sufficiently large constant $C_0$ (say at least $20d$).
	
	By combining the above two terms, we have
	\[
	\abs{\sum_{p\in P} \sigma(p) e^{-\norm{x-p}^2} - \sum_{p\in P} \sigma(p) e^{-\norm{s-p}^2} }  \leq  \frac{1}{5}D_i\cdot n_{i+1} e^{-I_i\norm{x}^2} \qedhere
	\]
	
\end{proof}

\leminductive*

\begin{proof}
	
	If $x\in B_\infty^d(n_{i+1})$, by Lemma \ref{lem:inductive_lipschitz}, we have 
	\[
	\abs{\sum_{p\in P} \sigma(p) e^{-\norm{x-p}^2} - \sum_{p\in P} \sigma(p) e^{-\norm{s-p}^2} } \leq \frac{1}{5}D_i\cdot n_{i+1} e^{-I_i\norm{x}^2}
	\]
	where $s\in S_i$ is the closest point to $x$ that $\abs{s_j}>\abs{x_j}$ for all $j=1,2,\dots,d$.
	We have, by the assumption on $\sigma$, 
	\begin{align*}
		\abs{\sum_{p\in P} \sigma(p) e^{-\norm{x-p}^2}} 
		& \leq 
		\abs{\sum_{p\in P} \sigma(p) e^{-\norm{s-p}^2} } + \frac{1}{5}D_i\cdot n_{i+1} e^{-I_i\norm{x}^2} \\
		& \leq
		C_1n_{i+1}e^{-\frac{2}{3}\norm{s}^2} + \frac{1}{5}D_i\cdot n_{i+1} e^{-I_i\norm{x}^2}.
	\end{align*}
	Since $I_i < \frac{2}{3}$ and $\norm{x} < \norm{s}$, the term $e^{-\frac{2}{3}\norm{s}^2}$ is less than  $e^{-I_i\norm{x}^2}$.
	Also, for $C> C_1$, we have $C_1 + \frac{1}{5}D_i < C + \frac{1}{5}D_i = C + C\cdot\frac{1}{4}(1-\frac{1}{5^i}) = D_{i+1}$.
	Hence, 
	\[
	\abs{\sum_{p\in P} \sigma(p) e^{-\norm{x-p}^2}}  < D_{i+1}n_{i+1}e^{-I_i\norm{x}^2} < D_{i+1}n_{i+1}e^{-I_{i+1}\norm{x}^2}
	\]
	
	If $x\notin B_\infty^d(n_{i+1})$, which means one of the coordinate $\abs{x_j} > n_{i+1}$, then, by the assumption on $\sigma$, 
	\[
	\abs{\sum_{p\in P} \sigma(p) e^{-\norm{x-p}^2}} \leq D_i\cdot n_{i} e^{-I_i\norm{x}^2}
	\]
	Note that $I_i = \frac{1}{3} + \frac{1}{3}(1-\frac{1}{2^{\ell(n)-i}}) = I_{i+1} + \frac{1}{3\cdot2^{\ell(n)-i}}$.
	It means that, recalling that $n_{i+1} = \sqrt{3\cdot 2^{\ell(n)-i}\log n_i}$,
	\[
	I_i\norm{x}^2 = (I_{i+1} + \frac{1}{3\cdot2^{\ell(n)-i}})\norm{x}^2 > I_{i+1}\norm{x}^2 + \frac{1}{3\cdot2^{\ell(n)-i}}n_{i+1}^2 = I_{i+1}\norm{x}^2 + \log n_i.
	\]
	Therefore,
	\begin{align*}
		\abs{\sum_{p\in P} \sigma(p) e^{-\norm{x-p}^2} }
		& < 
		D_i\cdot n_{i} e^{-I_i\norm{x}^2} 
		< 
		D_i\cdot n_{i} e^{-I_{i+1}\norm{x}^2 - \log n_i} \\
		& < 
		D_i e^{-I_{i+1}\norm{x}^2}
		<
		D_{i+1}n_{i+1}e^{-I_{i+1}\norm{x}^2} 
	\end{align*} 
	
\end{proof}

\subsection{Proofs in Section \ref{sec:full}}

\lemboundeddisc*

\begin{proof}
	
	First, we denote a $ (\sum_{i=0}^{\ell(n)-1}\abs{S_i}+\abs{P} )$-by-$ (\sum_{i=0}^{\ell(n)-1}\abs{S_i}+\abs{P} )$ matrix $M$ such that 
	\[
	M_{x,y} =  \left\{\begin{array}{lr}
		e^{-\norm{\sqrt{3}x - \sqrt{3}y}^2} & \text{for both $x,y\in P$}\\
		e^{-\norm{\frac{1}{\sqrt{3}}x - \sqrt{3}y}^2} & \text{for $x\in \cup_{i=0}^{\ell(n)-1}S_i$ and $y\in P$}\\
		e^{-\norm{\sqrt{3}x - \frac{1}{\sqrt{3}}y}^2} & \text{for $x\in P$ and $y\in \cup_{i=0}^{\ell(n)-1}S_i$}\\
		e^{-\norm{\frac{1}{\sqrt{3}}x - \frac{1}{\sqrt{3}}y}^2} & \text{for both $x,y\in \cup_{i=0}^{\ell(n)-1}S_i$}
	\end{array} \right.
	\]
	Since the Gaussian kernel is a positive-definite kernel, $M$ can be expressed as $V^TV$.
	We denote the columns of $V$ be $v^{(p)}$ for $p\in P$ and $w^{(s')}$ for $s'\in \cup_{i=0}^{\ell(n)-1}S_i$.
	Note that $\norm{v^{(p)}}$ and $\norm{w^{(s')}}$ are $1$.
	
	By Banaszczyk's Theorem, we have
	\allowdisplaybreaks
	\begin{align*}
		\MoveEqLeft \prob{\abs{\sum_{p\in P} \sigma(p)e^{-\norm{s'-p}^2 }} > C_1n_{i+1}e^{-\frac{2}{3}\norm{s'}^2} } \\
		& =
		\prob{\abs{\sum_{p\in P} \sigma(p)e^{-\frac{2}{3}\norm{s'}-\norm{\frac{1}{\sqrt{3}}s' - 
						\sqrt{3}p}^2+2\norm{p}^2}} > C_1n_{i+1}e^{-\frac{2}{3}\norm{s'}^2} }\\
		& =
		\prob{\abs{\sum_{p\in P}\sigma(p)\inner{ (
					\begin{matrix}
						0 \\
						w^{(s')}
					\end{matrix}
					)}{\frac{ (
						\begin{matrix}
							1 \\
							v^{(p)}e^{2\norm{p}^2}
						\end{matrix}
						)}{\sqrt{1+e^{4d}}}}} > \frac{C_1}{\sqrt{1+e^{4d}}}n_{i+1}} \\
		& \leq
		C'e^{-\frac{C''C_1^2n_{i+1}^2}{1+e^{4d}}}
	\end{align*}
	for any $s' \in S_i$ where $i=0,1,\dots,\ell(n)-1$.
	Recall that $C',C''$ are the absolute constants shown in Banaszczyk's Theorem.
	Note that $\abs{S_i} \leq (2C_0n_{i+1}n_i)^d < (2C_0n_i^2)^d$. 
	Therefore, if $C_1$ is a sufficiently large constant that is the same constant shown in Lemma \ref{lem:base} and Lemma \ref{lem:inductive}, we have 
	\begin{align*}
		\prob{ \abs{\sum_{p\in P} \sigma(p)e^{-\norm{s'-p}^2 }} > C_1n_1e^{-\frac{2}{3}\norm{s'}^2} } \leq
		C'(2C_0n_0^2)^de^{-\frac{C^{''}C_1^2n_{1}^2}{1+e^{4d}}} \leq \frac{1}{10}
	\end{align*}
	for all $s'\in S_0$ and also
	\begin{align*}
		\prob{ \abs{\sum_{p\in P} \sigma(p)e^{-\norm{s'-p}^2 }} > C_1n_{i+1}e^{-\frac{2}{3}\norm{s'}^2} } 
		& \leq
		C'(2C_0n_i^2)^de^{-\frac{C''C_1^2n_{i+1}^2}{1+e^{4d}}} \\
		& <
		C'(2C_0n_i^2)^d{n_{i}}^{-\frac{3C''C_1^22^{\ell(n)-i}}{1+e^{4d}}}
		< \frac{1}{10\cdot2^{\ell(n)-i}}
	\end{align*}
	for all $s'\in S_i$ where $i=1,2,\dots,\ell(n)-1$.
	Moreover, by Banaszczyk's Theorem, we can show
	\begin{align*}
		\prob{ \abs{\sum_{p\in P} \sigma(p)} > C } 
		& =
		\prob{ \abs{\sum_{p\in P}\sigma(p)\inner{e^{(1)}}{\frac{ (
						\begin{matrix}
							1 \\
							v^{(p)}e^{2\norm{p}^2}
						\end{matrix}
						)}{\sqrt{1+e^{4d}}}}} > \frac{1}{\sqrt{1+e^{4d}}}C } \\
		& \leq
		\frac{1}{10}
	\end{align*}
	Here, $e^{(1)}$ is a zero vector except that the first coordinate is $1$.
	Now, by Lemma \ref{lem:base} and inductively Lemma \ref{lem:inductive}, we conclude that 
	\[
	\abs{\sum_{p\in P} \sigma(p) e^{-\norm{x-p}^2} } < C\cdot \frac{5}{4}\cdot n_{\ell(n)} e^{-\frac{1}{3}\norm{x}^2}
	\]
	for all $x\in \mathbb{R}^d$ and 
	\[
	\abs{\sum_{p\in P} \sigma(p)} \leq C
	\]
	with probability at least $\frac{1}{2}$. \qedhere

\end{proof}

\lemboundeddiscbalance*

\begin{proof}
	
	From Lemma \ref{lem:bounded_disc}, we have 
	\[
	\abs{\sum_{p\in P} \sigma(p) e^{-\norm{x-p}^2} } < C\cdot \frac{5}{4}\cdot n_{\ell(n)} e^{-\frac{1}{3}\norm{x}^2}
	\]
	for all $x\in \mathbb{R}^d$ and 
	\[
	\abs{\sum_{p\in P} \sigma(p)} \leq C
	\]
	with probability at least $\frac{1}{2}$.
	
	Suppose there are more $+1$ than $-1$.
	Choose $\frac{1}{2}C$ points assigned $+1$ arbitrarily and flip them to $-1$ such that it makes the difference zero (or at most $1$ if $\abs{P}$ is an odd number).
	Denote $P_+ = \setdef{p\in P}{\sigma(p) = +1}$ and $P_- = \setdef{p\in P}{\sigma(p)  =-1}$.
	Also, $P_+'$ and $P_-'$ are defined in the same way after flipping the value.
	For any $x \in \mathbb{R}^d$,
	\begin{align*}
		\MoveEqLeft\abs{\sum_{p\in P_+'}e^{-\norm{x-p}^2}-\sum_{p\in P_-'}e^{-\norm{x-p}^2}} \\
		& \leq 
		\abs{\sum_{p\in P_+}e^{-\norm{x-p}^2}-\sum_{p\in P_-}e^{-\norm{x-p}^2}}+\abs{\sum_{p\in P_-'\cap P_+}e^{-\norm{x-p}^2}}
	\end{align*}
	We can bound the second term in the following way.
	\begin{align*}
		\MoveEqLeft \abs{\sum_{p\in P_-'\cap P_+}e^{-\norm{x-p}^2}} \\
		& =
		\abs{\sum_{p\in P_-'\cap P_+}e^{-\frac{1}{3}\norm{x}^2}e^{-\frac{2}{3}\norm{x-\frac{3}{2}p}^2}e^{\frac{1}{2}\norm{p}^2} } \\
		& \leq
		e^{-\frac{1}{3}\norm{x}^2}e^{\frac{1}{2}d} (\sum_{p\in P_-'\cap P_+}1 ) & \text{by $e^{-\frac{2}{3}\norm{x-\frac{3}{2}p}^2} \leq 1$ and $\norm{p}^2 \leq d$} \\
		& \leq
		\frac{1}{2}e^{-\frac{1}{3}\norm{x}^2}e^{\frac{1}{2}d}\cdot C
	\end{align*}
	Then, 
	\begin{align*}
		\MoveEqLeft\abs{\sum_{p\in P_+'}e^{-\norm{x-p}^2}-\sum_{p\in P_-'}e^{-\norm{x-p}^2}} \\
		& \leq 
		\abs{\sum_{p\in P_+}e^{-\norm{x-p}^2}-\sum_{p\in P_-}e^{-\norm{x-p}^2}}+\abs{\sum_{p\in P_-'\cap P_+}e^{-\norm{x-p}^2}} \\ 
		& \leq
		C\cdot \frac{5}{4}\cdot n_{\ell(n)} e^{-\frac{1}{3}\norm{x}^2} + \frac{1}{2}e^{-\frac{1}{3}\norm{x}^2}e^{\frac{1}{2}d}\cdot C \\
		& =
		O(n_{\ell(n)}e^{-\frac{1}{3}\norm{x}^2}) \qedhere
	\end{align*}

\end{proof}

\lemfull*

\begin{proof}
	
	From Lemma \ref{lem:bounded_disc_balance}, for each $Q_g$, we have 
	\[
	\abs{\sum_{p\in Q_g} \sigma_g(p) e^{-\norm{x-p}^2} } = O(n_{\ell(n)}e^{-\frac{1}{3}\norm{x-g}^2})
	\]
	for all $x\in \mathbb{R}^d$	with probability at least $\frac{1}{2}$.
	Since half of the points in $Q_g$ are $+1$, it means that half of the points in $P$ are also $+1$.
	
	Next, we argue that the final discrepancy $\abs{\mathcal{D}_{P,\sigma}(x)} = O(n_{\ell(n)})$ for all $x\in\mathbb{R}^d$.
	For the simpler computation, we loosen the bound as 
	\[
	\abs{\sum_{p\in Q_{g}}\sigma_g(p)e^{-\norm{x-p}^2}} = O(n_{\ell(n)}e^{-\frac{1}{3}\norm{x-g}_\infty^2})
	\]
	For any $x\in \mathbb{R}^d$ and an integer $r$, the $\ell_\infty$-ball centered at $x$ of radius $2r$ intersects at most $(2r+1)^d$ $\ell_\infty$-ball $Q_g$.
	It implies that the final discrepancy is 
	\begin{align*}
		\abs{\sum_{p\in P} \sigma(p) e^{-\norm{x-p}^2}} 
		& =
		\abs{\sum_{g\in \Gr_d(2)}\sum_{p\in Q_g} \sigma_g(p) e^{-\norm{x-p}^2}} \\
		& \leq
		\sum_{g\in \Gr_d(2)}\abs{\sum_{p\in Q_g} \sigma_g(p) e^{-\norm{x-p}^2}} \\
		& =
		\sum_{g\in \Gr_d(2)}O(n_{\ell(n)}e^{-\frac{1}{3}\norm{x-g}_\infty^2}) \\
		& =
		O (\sum_{r=1}^\infty \frac{(2r+1)^d}{e^{\frac{1}{3}(2(r-1))^2}} n_{\ell(n)}  ) = O(n_{\ell(n)})
	\end{align*}
	since the summation $\sum_{r=1}^\infty \frac{(2r+1)^d}{e^{\frac{1}{3}(2(r-1))^2}}$ converge to a constant, which depends on $d$ only.
	By Lemma \ref{lem:recurrence}, the final discrepancy is $\abs{\mathcal{D}_{P,\sigma}(x)} = O(1)$ for all $x\in\mathbb{R}^d$.

	We now analyze the running time.
	We first need to construct the set of input vectors $V_g$ in line \ref{line:construct_input}.
	It can be constructed by Cholesky decomposition and its running time is $O(\abs{Q_g}^3)$.
	We then analyze the running time of line \ref{line:main_step}.
	By Lemma \ref{lem:bounded_disc_balance}, for each $Q_g$, the algorithm in Theorem \ref{thm:b_thm} constructs a coloring that has the discrepancy $O(n_{\ell(n)}e^{-\frac{1}{3}\norm{x}^2})$ with probability $1/2$.
	Therefore, the expected number of executions of line \ref{line:main_step} is $2$.
	Recall that the running time of the algorithm in Theorem \ref{thm:b_thm} is $\abs{Q_g}^{\omega+1}$ where $\omega$ is the exponent of matrix multiplication.
	Also, it takes $O ( \abs{Q_g} \cdot\abs{\cup_{i=0}^{\ell(n)-1}S_{i}}  ) = O ( \abs{Q_g} \cdot \log^{O(d)}\abs{Q_g}  ) $ to verify the discrepancy in line \ref{line:check} by Lemma \ref{lem:base} and Lemma \ref{lem:inductive}.
	Finally, by assuming $d$ is constant, the total expected running time is $O(n^{\omega+1})$. \qedhere
	
\end{proof}

\end{document}